\documentclass[10pt]{article}
\usepackage{amsmath,amssymb,amsthm}
\usepackage{geometry}
\usepackage{fancyhdr}
\usepackage{setspace}
\usepackage{hyperref}
\usepackage{authblk}
\usepackage{graphicx}
\usepackage{enumitem} 
\usepackage{titlesec}
\usepackage{extpfeil}

\usepackage{xcolor}

\titleformat{\section}{\large\bfseries}{\thesection.}{1em}{}
\titleformat{\subsection}{\large\itshape}{\thesubsection.}{1em}{}
\titleformat{\subsubsection}{\large\itshape}{\thesubsubsection.}{1em}{}
\newtheorem{theorem}{Theorem} 
\newtheorem{lemma}{Lemma}
\newtheorem{proposition}{Proposition}
\newtheorem{corollary}{Corollary}
\newtheorem{definition}{Definition}

\newtheorem{assumption}{Assumption}
\newtheorem{remark}{Remark}
\newtheorem{axiom}{Axiom}

\newcommand{\keywords}[1]{%
  \vspace{1em}\par\noindent\textbf{Keywords:} #1\par\vspace{1em}%
}

\def \o{\omega}
\def \O{\Omega}
\title{\huge Robust No-Arbitrage under Projective Determinacy}
\author[2]{Alexandre Boistard}
\author[1]{Laurence Carassus}
\author[2]{Safae Issaoui}
\affil[1]{MICS, Centrale-Sup\'{e}lec, Universit\'{e} Paris-Saclay, France}
\affil[2]{Dominante MDS, Centrale-Sup\'{e}lec, Universit\'{e} Paris-Saclay, France}
\date{\vspace{0em}\today}

\begin{document}
\maketitle

\begin{abstract}
Drawing on set theory, this paper contributes to a deeper understanding of the structural condition of mathematical finance under Knightian uncertainty.

We adopt a projective framework in which all components of the model -- prices, priors and trading strategies -- are treated uniformly in terms of measurability. This contrasts with the quasi-sure setting of Bouchard and Nutz, in which prices are Borel-measurable and graphs of local priors are analytic sets, while strategies and stochastic kernels inherit only universal measurability.

In our projective framework, we establish several characterizations of the robust no-arbitrage condition, already known in the quasi-sure setting, but under significantly more elegant and consistent assumptions. 
These characterisations have important applications, in particular, the existence of solutions to the robust utility maximization problem.

To do this, we work within the classical Zermelo-Fraenkel set theory with the Axiom of Choice (ZFC), augmented by the axiom of Projective Determinacy (PD). The (PD) axiom, a well-established axiom of descriptive set theory, guarantees strong regularity properties for projective sets and projective functions.

\end{abstract}

\keywords{Robust Finance, Quasi-sure No-Arbitrage, Projective Determinacy, Projective set, Projective function.}

\section{Introduction}
The no-arbitrage hypothesis is a cornerstone in financial mathematics and economic theory, ensuring the internal consistency of pricing models, optimal solutions in portfolio selection models and preventing arbitrage opportunities that could destabilise markets. The no-arbitrage principle asserts that making a non-risky profit with zero net investment is impossible. 
Traditional approaches assume a single probability measure to describe the evolution of asset prices; however, in a multiple-priors (or robust or Knightian) framework, uncertainty is modelled through a family of probability measures or a set of events. This generalization accounts for ambiguity and model uncertainty, making it particularly relevant in modern financial markets where agents may hold diverse and even conflicting beliefs about future states of the world. 
The earliest literature assumed that the set of beliefs is dominated. We refer to \cite{FSW09} for a comprehensive survey of the dominated case. 
Unfortunately, this setting excludes volatility uncertainty and is easily violated in discrete time (see \cite{blanchard2020}); this is why we focus on the non-dominated case.

Different notions of arbitrage have been developed in discrete-time robust finance. 
The quasi-sure no-arbitrage condition of Bouchard and Nutz (\cite{bouchard2015}) $NA(\mathcal{Q})$ states that if the terminal value of a trading strategy, starting from 0, is non-negative $\mathcal{Q}$-quasi-surely, then it always equals $0$ $\mathcal{Q}$-quasi-surely, where $\mathcal{Q}$ represents all the possible probability measures or beliefs. 
Here $\mathcal{Q}$-quasi-surely roughly means $P$-a.s. for all $P\in \mathcal{Q}$. 
The pathwise approach takes a scenario-based interpretation of arbitrage rather than relying on a set of probability measures: a subset of relevant events or scenarios without specifying their relative weight is given (see, for example, \cite{Burz18}). Notably, \cite{OW18} have unified the quasi-sure and the pathwise approaches, demonstrating, under certain regularity assumptions, that both approaches are equivalent. We also mention the model-independent approach, discussed, for example, in \cite{Rie15}. 

Here, we focus on the quasi-sure no-arbitrage condition of Bouchard and Nutz, which has become dominant in the discrete-time literature. 
However, under this condition, it is not even clear if there exists a belief $P \in \mathcal{Q}$ satisfying the single-prior no-arbitrage condition $NA(P)$. It is indeed true, but $\mathcal{Q}$ might still contain some models that are not arbitrage-free (see \cite{blanchard2020}). In \cite{blanchard2020}, the authors have shown that the $NA(\mathcal{Q})$ condition is equivalent to the existence of a subclass of priors $\mathcal{P}\subseteq \mathcal{Q}$ such that $\mathcal{P}$ and $ \mathcal{Q}$ have the same polar sets (roughly speaking, the same relevant events) and $NA(P)$ holds for all $P \in \mathcal{P}$. So instead of $NA(\mathcal{Q})$, one may assume that every model in $\mathcal{P}$ is arbitrage-free. Under quasi-sure uncertainty, these perspectives provide a more flexible framework for pricing and hedging.  It also allows tractable theorems for the existence of solutions to the problem of robust utility maximisation (see  \cite{blanchard2020},  \cite{BCK19} or \cite{RaMe18}). The construction of $\mathcal{P}$ is based on the 
existence of a probability measure for which the single-prior no-arbitrage condition holds $\mathcal{Q}$-quasi-surely and the affine hull of the price increments support is equal to the quasi-sure one, again $\mathcal{Q}$-quasi-surely.

In the framework of Bouchard and Nutz, random sets of local priors are first given. These probability measures are ``local" in that they represent the investor's belief between two successive moments. The cornerstone assumption of \cite{bouchard2015} is that the graphs of these random sets are analytic sets. 
Thanks to this assumption and to the measurable selection theorem of Jankov-von Neumann (see \cite{Bertsekas1978}), it is possible to obtain local beliefs that are analytically and, thus, universally measurable, as a function of the path.  The intertemporal set of beliefs can then be constructed from these kernels as product measures. Measurable selection is also necessary to do the way back, for example, to go from intertemporal quasi-sure inequalities to local quasi-sure ones, as when going from intertemporal no-arbitrage to local ones. For that, Bouchard and Nutz rely on 
 the uniformization of nuclei of Suslin schemes on the product of the universal sigma-algebra and the Borel one, as discussed by Leese in \cite{Leese78}.  So, one needs to go outside the class of analytic sets (which are the nuclei of Suslin schemes on the Borel sigma-algebra).  
Moreover, Bouchard and Nutz use upper semianalytic functions. A technical issue is that the composition of two upper semianalytic functions may not remain upper semianalytic. This is why the prices are assumed to be Borel measurable. Furthermore, the class of analytics sets is not closed under complement. For example, the set where the local quasi-sure no-arbitrage holds is co-analytic, and if we restrict upper semianalytic functions to this set, they are no longer upper semianalytic. 
Summing up in the classical framework of Bouchard and Nutz, the price processes are assumed to be Borel measurable, the graphs of random beliefs to be analytic sets, while trading strategies are only obtained to be universally measurable. 
 The conditions of measurability are not homogeneous, and you have to assume a lot (Borel, analytical sets) to obtain little (universally measurable).

To address this issue, an interesting development is the connection between robust finance and advanced set-theoretic axioms. 
Projective Determinacy, an axiom from descriptive set theory, has emerged as a powerful tool when dealing with Knightian uncertainty (see \cite{burzoni2019} and \cite{CF24}). 
Other examples in mathematical finance and economics where some set-theoretic axioms are used can be found in \cite{refmai_pj} and \cite{BCK19}. 
Using projective sets instead of analytic sets or nuclei of Suslin schemes has been particularly fruitful in handling non-dominated model uncertainty, especially in non-concave utility maximisation. Assuming the axiom of Projective Determinacy, projective sets share the same regularity properties as analytic sets.  They are also stable by complement, and the composition of projective functions remains projective. 
Projective Determinacy is rooted in early 20th-century mathematical logic, with contributions from various mathematicians and logicians. 
In the 1980s, significant advances were made by Martin and Moschovakis (see \cite{refboreldet_pj} and the textbook of Moschovakis \cite{Moschovakis}) and then by Woodin with the connection to the existence of large cardinals (see \cite{refwood1_pj} for a survey).  
Determinacy refers to the existence of a winning strategy for one of the two players of an infinite game, and the axiom of Projective Determinacy states that every projective set is determined. 



We use the same projective framework as introduced by Carassus and Ferhoune for solving robust non-concave utility maximization (see \cite{CF24}). The price processes are assumed to be projectively measurable, the graphs of random beliefs to be projective sets, and we obtain projectively measurable stochastic kernels and trading strategies. 
This allows for simplifying and unifying the technical assumptions (and the proofs) that were needed in previous works, such as analytic and Borel measurability. 
We characterise the quasi-sure no-arbitrage condition in the projective setup. 
We show that the $NA(\mathcal{Q})$ condition is equivalent to the existence of $P^*\in \mathcal{Q}$ for which the geometric form of $NA(P^*)$ holds $\mathcal{Q}$-quasi-surely and such that the affine hull of the price increments support under $P^*$ is equal to the quasi-sure one, again $\mathcal{Q}$-quasi-surely. We also show that $NA(\mathcal{Q})$ is equivalent to the existence of $\mathcal{P}\subseteq \mathcal{Q}$ such that $\mathcal{P}$ and $ \mathcal{Q}$ have the same polar sets and $NA(P)$ holds for all $P \in \mathcal{P}$. Note that in \cite{CF24}, $P^*$ was a key tool to prove the existence of projectively measurable solutions to the robust non-concave utility maximisation problem. But the equivalence between the existence of $P^*$ and the $NA(\mathcal{Q})$ condition was only conjectured. We therefore provide a positive answer to the conjecture made by Carassus and Ferhoune. One further advantage of the projective framework is that the single-prior models naturally appear as a special case. Indeed, in the model of Bouchard and Nutz, one must assume that the graph of this prior is analytic, whereas in the projective setting, this graph is automatically projective. As a consequence, the no-arbitrage characterization in markets with a single prior emerges as a direct corollary. 

Our proofs rely on properties of projective functions and sets proved in \cite{carassus2024}. We also establish a key result that allows us to go from $\mathcal{Q}$-quasi-sure positivity to local positivity on a projective set of full measure. This result is inspired by Lemma A.1 of \cite{Carassus25}, which was formulated in the context of nuclei of Suslin schemes on the product of the universal sigma-algebra and the Borel one. 
 This allows us to prove that the quasi-sure no-arbitrage condition is consistent with its local version at each time step, and to work initially in a one-period model.
 In addition, several proofs are adaptations of arguments from \cite{blanchard2020} to the projective framework. We also provide a complete one-period proof of the characterization of $NA(\mathcal{Q})$.


To summerize, the interplay between Projective Determinacy and measurable selection provides a powerful foundation for understanding dynamic decision-making in ambiguous environments, reinforcing the theoretical underpinnings of multiple-priors financial models.



The paper is structured as follows: Section \ref{projsetup} explains the projective setup, while Section \ref{finsetup} presents the financial setting.    Section \ref{mainresult} contains our main results, while Section \ref{proofs} presents their proofs. Finally, the Appendix contains results about the one-period model and properties of projective sets and functions.\\
We finish this introduction with some notations and definitions.  
For all Polish spaces $X$, we denote by $\mathfrak{P}(X)$ the set of probability measures defined on the measurable space $(X, \mathcal{B}(X))$, where $\mathcal{B}(X)$ is the Borel sigma-algebra on $X$. We define the universal sigma-algebra on $X$ as
\[
\mathcal{B}_c(X) := \bigcap_{P \in \mathfrak{P}(X)} \mathcal{B}_P(X),
\]
where $\mathcal{B}_P(X)$ denotes the completion of $\mathcal{B}(X)$ with respect to $P \in \mathfrak{P}(X)$. For the rest of this paper, we use the same notation for $P \in \mathfrak{P}(X)$ and its (unique) extension on $\mathcal{B}_c(X)$.

In this context, a set $A$ is called $\mathcal{Q}$-polar for some $\mathcal{Q} \subseteq \mathfrak{P}(X)$ if there exists $N \in \mathcal{B}_c(X)$ such that $A\subseteq N$ and $P[N] = 0$ for all $P \in \mathcal{Q}.$ Moreover, a set $B$ is of $\mathcal{Q}$-full-measure if $X\setminus B$ is $\mathcal{Q}$-polar. Finally, a property holds true $\mathcal{Q}$-quasi-surely (q.s.) if it holds true on a  $\mathcal{Q}$-full-measure set. 

Let $x \in \mathbb{R}^d$ and $\epsilon >0,$ then $|x|$ is the Euclidian norm of  $x$ and  the open ball of center $x$ and radius $\epsilon$ is denoted by $B(x,\epsilon):=\{y\in\mathbb{R}^d:|y-x|<\epsilon\}$. 

\section{Projective setup}
\label{projsetup}
We introduce our projective setup.

\begin{definition}[Projective sets] \label{def_projective_sets}
Let $X$ be a Polish space. 
An analytic set of $X$ is the projection into $X$ of a Borel subset of $X \times \mathbb{N}^\mathbb{N}$. The class of such sets is denoted by $\Sigma^1_1(X)$. The complement of an analytic set is called a co-analytic set, which class is denoted by $\Pi^1_1(X)$. 

For $n \geq 2$, the classes of analytic and co-analytic sets of order $n$ are defined recursively:
\[
\Sigma^1_{n}(X) := \bigl\{ \mathrm{proj}_X(C) : C \in \Pi^1_{n-1}(X \times \mathbb{N}^\mathbb{N})\bigr \}, \quad \Pi^1_{n}(X) :=\bigl \{ X \setminus C : C \in \Sigma^1_{n}(X) \bigr\}.
\]
For all $n\geq 1$, the intersection of these two classes defines $\Delta^1_n(X)$:
   \[
   \Delta^1_n(X) := \Sigma^1_n(X) \cap \Pi^1_n(X).
   \]
   
Finally, the class of projective sets on $X$ is defined as
\[
\mathbf{P}(X) := \bigcup_{n \geq 1} \Delta^1_n(X).
\]

\end{definition}
Note that the only sets which are analytic and co-analytic are the Borel sets, see \cite[Theorem 14.11, p88]{kechris1995}: 
\begin{equation}\label{equation_borel_set}
    \mathcal{B}(X)=\Sigma^1_{1}(X)\cap\Pi^1_{1}(X)=\Delta^1_1(X). 
\end{equation}

Let $X$ be a Polish space. We now define the notion of measurability used in this paper.
\begin{definition}[Projective functions] \label{def_projective_mes_functions}
A function $f: X \to \mathbb{R}^d $ is $\Delta_n^1(X) $-measurable if $f^{-1}(B) $ belongs to $\Delta_n^1(X) $ for all Borel sets $B \subseteq \mathbb{R}^d $. The function $f $ is projective (or projectively measurable or $\mathbf{P}(X) $-measurable), if there exists $n \in \mathbb{N}^*$ such that $f $ is $\Delta_n^1(X) $-measurable.
\end{definition}
For two  Polish spaces $X$ and $Y$, we will denote set-valued mappings\footnote{A set-valued mapping $F: X \twoheadrightarrow Y$ is a mapping such that for every $x \in X$, $F(x)$ is a subset of $Y$.} as $F:X \twoheadrightarrow Y$. 
\begin{definition}[Projective mappings] \label{def_projective_mes_mapping}
A set-valued mapping $F: X \twoheadrightarrow \mathbb{R}^d $ is $\Delta_n^1(X) $-measurable if $F^{-1}(O) := \{ x \in X : F(x) \cap O \neq \emptyset \} $ belongs to $\Delta_n^1(X) $ for all open sets $O \subseteq \mathbb{R}^d $ (see also \cite[Definition 14.1, p.643]{RockafellarWets1998}).
The mapping $F $ is projective (or projectively measurable or $\mathbf{P}(X) $-measurable), if there exists $n \in \mathbb{N}^* $ such that $F $ is $\Delta_n^1(X) $-measurable.
\end{definition}

\begin{remark}\label{n_for_projection}
     The $n\in \mathbb{N}^*$ defined in these two definitions is independent of the Borel or open sets and only depends on the function or set-valued mapping.
\end{remark}

We need the notion of determined sets to state the (PD) axiom. Fix a set $A \subseteq \mathbb{N}^\mathbb{N}$. Consider a two-player infinite game. Player I plays $a_0 \in \mathbb{N} $, then Player II plays $b_0 \in \mathbb{N} $, then Player I plays $a_1 \in \mathbb{N} $, etc. A play is a sequence $(a_0, b_0, a_1, b_1, \dots) \in \mathbb{N}^{\mathbb{N}} $. Player I wins the game if $(a_0, b_0, a_1, b_1, \dots) \in A $. Otherwise, if $(a_0, b_0, a_1, b_1, \dots) \in \mathbb{N}^{\mathbb{N}} \setminus A $, Player II wins. A winning strategy for a Player is a strategy under which the Player always wins; that is, the result of the game always belongs to the set $A$ for Player I or to $\mathbb{N}^\mathbb{N}\setminus A$ for Player II, regardless of what the other Player plays.
\begin{definition}[Determined sets]
\label{def_determined_sets}
    A set $A \subseteq \mathbb{N}^\mathbb{N}$ is determined if a winning strategy exists for one of the two players.
\end{definition}
 Borel sets (see \cite{refboreldet_pj}) are determined. However, the determinacy of Borel sets is the best possible result provable in ZFC. This is why Martin and Moschovakis independently introduce the axiom of Projective Determinacy (see, for example, \cite[Definition 38.15, p. 325]{kechris1995} in the textbook of Kechris).
\begin{axiom}[Projective Determinacy]\label{axiom_PD}
The Projective Determinacy (PD) axiom states that if $A \subseteq \mathbb{N}^\mathbb{N}$ is a projective set, then $A$ is determined.
\end{axiom}
\begin{remark}
The (PD) axiom is a fruitful axiom as it allows to answer old questions on projective sets coming from Lusin. What makes this axiom plausible is that it is implied by other axioms of descriptive set theory that, a priori, have no direct connection with projective sets. In particular, the (PD) axiom is closely connected to the existence of Woodin cardinals. Martin and Steel (see \cite{reflargedet_pj}) showed that if sufficiently many Woodin cardinals exist, then all projective sets are determined. Conversely, assuming the (PD) axiom, one can construct inner models (transitive set-theoretic classes that satisfy the axioms of ZF and contain all the ordinals) which exhibit sequences of Woodin cardinals. 
The relative consistency of large cardinals with ZFC has been extensively studied via inner model theory and forcing techniques, and no contradictions are known assuming that ZFC itself is consistent (see \cite{refwood1_pj}). 
Quoting Woodin \lq\lq{}Projective Determinacy is the correct axiom for the projective sets; the ZFC axioms are obviously incomplete and, moreover, incomplete in a fundamental way.\rq\rq{}
\end{remark}
We will not apply the (PD) axiom directly, but rather one of the reasons why it was introduced: it implies that projective sets are universally measurable and that a projectively measurable selection is possible on projective sets. 
\begin{proposition}[Consequences of the (PD) axiom]\label{prop_conseq_PD_axiom}
Assume the (PD) axiom.
    \begin{enumerate}
    \item[(i)] If $X$ is a Polish space, then $\mathbf{P}(X)\subseteq\mathcal{B}_c(X)$. 
    \item[(ii)]  Let $X $ and $Y $ be Polish spaces and $A \in \mathbf{P}(X \times Y) $. Then, there exists a projective function $\phi : \operatorname{proj}_X(A) \to Y $ such that $\operatorname{Graph}(\phi) \subseteq A $.
    \end{enumerate}

\end{proposition}
\begin{proof}
    See \cite[Theorem 38.17, p.~326]{kechris1995} and \cite[Proposition 6]{carassus2024}.
\end{proof}

\section{Financial setting}
\label{finsetup}

We fix a time horizon $T $ and introduce a family of Polish spaces $(\Omega_t)_{t \in \{1, \dots, T\}} $. 
For all $t \in \{0, \dots, T\} $, let $\Omega^t := \Omega_1 \times \dots \times \Omega_t$ with the convention that $\Omega^0=\{\omega_0\}$ is a singleton. For all $t \in \{0, \dots, T\},$  let $S_t: \O^t \to \mathbb{R}^d$. Then, $S := (S_t)_{t \in \{0, \dots, T\}} $ is the $\mathbb{R}^d$-valued process representing the price of the $d $  risky assets over time. A riskless asset whose price equals 1 is also available. We are now in a position to state our first assumption.

\begin{assumption}[Measurability of the prices]\label{assumption_prices}
For all $t \in \{0,\dots,T\} $, $S_t $ is $\mathbf{P}(\Omega^t)$-measurable. In the case $t = 0$, we mean that $S_0$ is a constant.
\end{assumption}
\begin{definition}[Trading strategies]\label{def_strategy} 
Let $\phi_t: \O^{t-1} \to \mathbb{R}^d$ for all $t \in \{1, \dots, T\}.$  
A trading strategy $\phi$ is a $d$-dimensional process $\phi := \{\phi_t : t \in \{1, \dots, T\}\}$ such that $\phi_t$ is $\mathbf{P}(\Omega^{t-1})$-measurable for all $t \in \{1, \dots, T\}$. In the case $t = 1$, we mean that  $\phi_1$ is a constant. We denote by $\Phi$ the set of such strategies, which are also self-financing.
\end{definition}

For $x,y\in \mathbb{R}^d,$ the scalar product of $x$ and $y$ will be concatenated as $xy$. For $\phi \in \Phi$, $V_t^{x, \phi}$ denotes the value of the strategy $\phi$ at time $t \in \{0, \dots , T\}$ with initial endowment  $x \in \mathbb{R}$. We get that
\[
V_t^{x, \phi} = x + \sum_{s=1}^{t} \phi_s \Delta S_s.
\]

We now construct the set $\mathcal{Q}^T$ of all prevailing priors. The set $\mathcal{Q}^T$ captures all the investor’s beliefs about the law of nature. It is construct out of the one-step priors $\mathcal{Q}_{t+1} : \Omega^t \twoheadrightarrow \mathfrak{P}(\Omega_{t+1})$ where $\mathcal{Q}_{t+1}(\omega^t)$ is the set of all possible priors for the $(t+1)$-th period given the state $\omega^t$ at time $t$, for all $t \in \{0, \dots , T-1\}$. The following assumption allows us to perform measurable selection (see Proposition~\ref{prop_conseq_PD_axiom}).
\begin{assumption}[Measurability of the beliefs]
\label{assumption:graph_Q}
For all $t \in \{0,\dots,T-1\} $, $\mathcal{Q}_{t+1} : \Omega^t \twoheadrightarrow \mathfrak{P}(\Omega_{t+1})$ is a nonempty and convex-valued random set such that
\[
\mathrm{Graph}(\mathcal{Q}_{t+1}):= \bigl\{ (\omega^t, P) \in \Omega^t \times \mathfrak{P}(\Omega_{t+1}): P \in \mathcal{Q}_{t+1}(\omega^t) \bigr \}
\in \mathbf{P}\bigl(\Omega^t \times \mathfrak{P}(\Omega_{t+1})\bigr).\] 
In the case $t = 0$, we mean that $\mathcal{Q}_{1} =\mathcal{Q}_{1}(\omega^0) $ is a nonempty and convex (nonrandom) set of $\mathfrak{P}(\Omega_{1})$.
\end{assumption}

We set $SK_1=\mathfrak{P}(\Omega_{1})$. Let $t \in \{1, \ldots, T-1\}$ and $q_{t+1}(\cdot \mid \cdot) : \mathcal{B}(\Omega_{t+1}) \times \Omega^t \to \mathbb{R}$.  We say that $q_{t+1} \in SK_{t+1}$ if for all $\omega^t \in\Omega^t $, $q_{t+1}(\cdot \mid \omega^t) \in \mathfrak{P}(\Omega_{t+1})$ and $\Omega^t \ni \omega^t\mapsto q_{t+1}(\cdot \mid \omega^t)\in \mathfrak{P}(\Omega_{t+1})$ is projectively measurable. So, $SK_{t+1}$ is the set of projectively measurable stochastic kernels on $\Omega_{t+1}$ given $\Omega^t$. 


\begin{remark}[About Assumptions]
    In the setting of Bouchard and Nutz, $S_t$ is assumed to be Borel measurable and $\mathrm{Graph}(\mathcal{Q}_{t+1})$ to be an analytic set. As without the (PD) axiom, Borel functions are projective (choose $n=1$ in Definition~\ref{def_projective_mes_functions} and recall~(\ref{equation_borel_set})), and as analytic sets are projective sets (see Definition~\ref{def_projective_sets}), our assumptions are thus weaker in ZFC. 
    Under the (PD) axiom, if $\phi \in \Phi$, then $\phi$ is universally measurable (see Proposition~\ref{prop_conseq_PD_axiom}), which is the usual assumption in the quasi-sure literature. The same reasoning holds for stochastic kernels. So, our assumptions are again weaker, but we are assuming the (PD) axiom this time.
\end{remark}

Under the (PD) axiom and Assumption~\ref{assumption:graph_Q}, Proposition~\ref{prop_conseq_PD_axiom} allows us to perform measurable selection on $\mathrm{Graph}(\mathcal{Q}_{t+1}) \in \mathbf{P}(\Omega^t \times \mathfrak{P}(\Omega_{t+1}))$ and we obtain that there exists $q_{t+1} \in SK_{t+1}$ such that for all $\omega^t \in \Omega^t$ (recall that $\mathrm{proj}_{\Omega^t}\, \mathrm{Graph}(\mathcal{Q}_{t+1})=\{\mathcal{Q}_{t+1} \neq \emptyset\}=\Omega^t$ from Assumption~\ref{assumption:graph_Q}), $q_{t+1}(\cdot \mid \omega^t) \in \mathcal{Q}_{t+1}(\omega^t)$.

Now, for all $t \in \{1, \ldots, T\}$, there exists (see Remark~\ref{remark_integrals}) a unique product measure $q_1 \otimes \cdots \otimes q_t$ which belongs to $\mathfrak{P}(\Omega^t)$ and is such that for all $A^t:=A_1\times\dots\times A_t \in\Omega^t$:
\[
q_1 \otimes \cdots \otimes q_t[A^t] := \int_{A_1} \dots \int_{A_t} q_t\bigl(d\omega_t \mid (\omega_1, \dots, \omega_{t-1})\bigr) \dots q_1(d\omega_1).
\]
We are in position to define our intertemporal sets of priors $\mathcal{Q}^t \subseteq \mathfrak{P}(\Omega^t)$ for all $t \in \{1, \ldots, T\}$ by:
\[
\mathcal{Q}^t := \bigl\{q_1 \otimes q_2 \otimes \cdots \otimes q_t : q_{s+1} \in SK_{s+1},\, q_{s+1}(\cdot \mid \omega^s) \in Q_{s+1}(\omega^s), \, \forall \omega^s \in \Omega^s, \, \forall s \in \{0, \ldots, t-1\} \bigr\}.
\]
We also set $\mathcal{Q}^0 := \{\delta_{\omega_0}\}$, where $\delta_{\omega_0}$ is the Dirac measure on the single element $\omega_0$ of $\Omega^0$. If $P := q_1 \otimes q_2 \otimes \cdots \otimes q_T \in Q^T$, we write for any $t \in \{1, \ldots, T\}$, $P^t := q_1 \otimes q_2 \otimes \cdots \otimes q_t$ and $P^t \in \mathcal{Q}^t$. In this paper, we mostly work directly on the disintegration of $P$ rather than $P$.

\begin{remark}[Integrals, product measures]\label{remark_integrals}
   Let $X $ be a Polish space. Let $f : X \to \mathbb{R} \cup \{-\infty, +\infty\} $ be a universally measurable function and let $p \in \mathfrak{P}(X) $. We define the $(- \infty)$ integral denoted by $\int_- f dp$ and the $(+\infty)$ integral denoted by $\int^- f dp$ as follows. When $\int f^+ dp < +\infty$ or $\int f^- dp < +\infty$, both integrals are equal and are defined as the extended integral of $f $:
    \[  
     \int_- f dp = \int^- f dp :=\int f^+ dp - \int f^- dp.
    \]
Otherwise, $\int_- f dp := -\infty$ and $\int^- f dp := +\infty$. The condition of no-arbitrage is mainly applied in the context of maximization problems (portfolio choice, super-replication). That's why we use the $\int _- $ integral and the associated arithmetic rule $\infty-\infty=-\infty+\infty=-\infty.$ In the rest of the paper, we simply denote $\int_- f dp $ by $\int f dp$ if no further precision is necessary.

   We have seen in Proposition~\ref{prop_conseq_PD_axiom} that under the (PD) axiom, any projective set $A$ is universally measurable. This allows us to define $p[A]$ for any probability measure $p$ and, more generally, to use classical measure theory results in the projective context. First, any projective function $f$ is universally measurable (see Proposition~\ref{prop_conseq_PD_axiom}) so that $\int f dp$ (as defined above) is well-defined. Moreover, it is possible to construct a unique probability measure on the product space from projectively measurable stochastic kernels and also to use Fubini\rq{}s theorem when $f$ is projective (see \cite[Proposition~7.45~p.175]{Bertsekas1978}). So, the sets $(Q^t)_{t \in \{0,\dots,T-1\}}$ are indeed well-defined.

\end{remark}

\begin{definition}[Supports]\label{def_conditional_supports}
    Let $t \in \{0,\dots,T-1\} $ and $P \in \mathcal{Q}^T $ with the fixed disintegration 
   $
    P:= q_1 \otimes \dots \otimes q_T.
    $
    The random sets $E^{t+1}: \Omega^t \times \mathfrak{P}(\Omega_{t+1}) \twoheadrightarrow \mathbb{R}^d $, $D^{t+1}: \Omega^t \twoheadrightarrow \mathbb{R}^d $ and $D_P^{t+1}: \Omega^t \twoheadrightarrow \mathbb{R}^d $ are defined by
    \begin{align*}
    E^{t+1}(\omega^t, p) &:= \bigcap \bigl\{ A \subseteq \mathbb{R}^d:\text{closed, } p[ \Delta S_{t+1}(\omega^t, \cdot) \in A ] = 1 \bigr\}, \\
    D^{t+1}(\omega^t) &:= \bigcap \{ A \subseteq \mathbb{R}^d :\text{closed, } p[\Delta S_{t+1}(\omega^t, \cdot) \in A]= 1, \ \forall p \in \mathcal{Q}_{t+1}(\omega^t) \}, \\
    D^{t+1}_P(\omega^t) &:= \bigcap \{ A \subseteq \mathbb{R}^d : \text{closed, } q_{t+1}(\Delta S_{t+1}(\omega^t, \cdot) \in A \mid \omega^t) = 1 \}.
    \end{align*}
    We call $D^{t+1}$ the quasi-sure support of $\Delta S_{t+1}$ and  $D^{t+1}_P$ the support of $\Delta S_{t+1}$ relatively to $P$.
\end{definition}
If $R\subseteq \mathbb{R}^d$, $\mathrm{Aff}(R)$ denotes the smallest affine set containing $R$, $\mathrm{Conv}(R)$ denotes the smallest convex set containing $R$ and if $R$ is convex, $\mathrm{Ri(R)}$ is the interior of $R$ relatively to $\mathrm{Aff}(R)$. 

\begin{remark}\label{EP_included_D}
    For all $\omega^t \in \Omega^t$ and all $p \in \mathcal{Q}_{t+1}(\omega^t)$, $E^{t+1}(\omega^t, p) \subseteq D^{t+1}(\omega^t)$. Indeed, let $p \in \mathcal{Q}_{t+1}(\omega^t)$. 
Using the one-period result  \cite[Lemma 4.2]{bouchard2015} (which just requires that $\mathcal{Q}_{t+1}(\omega^t)$ is a nonempty set of probability measures), we get that $D^{t+1}(\o^{t})$ is a closed set of $\mathbb{R}^d$
such that   
$p[\Delta S_{t+1}(\o^{t},\cdot) \in D^{t+1}(\o^{t})]=1$ (it is even the smallest). By definition of $E^{t+1}(\omega^t, p)$ as an intersection of such  sets, $E^{t+1}(\omega^t, p)\subseteq D^{t+1}(\omega^t)$. 
    
    If $P:=q_1\otimes\dots\otimes q_T\in \mathcal{Q}^T$, then for all $t\in \{1,\dots,T-1\}$ and $\omega^t\in \Omega^t$, $D^{t+1}_P(\omega^t)=E^{t+1}(\omega^t,q_{t+1}(\cdot\mid \omega^t))$.

\end{remark}

We now introduce the definitions of no-arbitrage.

\begin{definition}[Quasi-sure no-arbitrage condition]
The condition $NA(\mathcal{Q}^T)$ holds true if
$$
\quad V_T^{0, \phi} \geq 0 \; \mathcal{Q}^T\mbox{-q.s.} \text{ for some }  \phi \in \Phi \implies V_T^{0, \phi} = 0 \; \mathcal{Q}^T\mbox{-q.s.}
$$

\end{definition}
\begin{definition}[Single-prior no-arbitrage condition]
Let $P \in \mathfrak{P}(\Omega^T).$ The condition $NA(P)$ holds true if
$$
\quad V_T^{0, \phi} \geq 0 \; P\mbox{-a.s.} \text{ for some }  \phi \in \Phi  \implies V_T^{0, \phi} = 0 \; P\mbox{-a.s.}
$$

\end{definition}
Note that when $\mathcal{Q}^T=\{P\}$ both notions coincide. 
\begin{definition}[Local no-arbitrage condition]
Fix $t\in \{0,\dots,T-1\}$ and $\omega^t\in\Omega^t$. The condition $NA(\mathcal{Q}_{t+1}(\omega^t))$ holds true if
\[
y\Delta S_{t+1}(\omega^t,\cdot) \geq 0 \; \mathcal{Q}_{t+1}(\omega^t)\text{-q.s. for some } y\in \mathbb{R}^d \implies y\Delta S_{t+1}(\omega^t,\cdot) = 0 \; \mathcal{Q}_{t+1}(\omega^t)\text{-q.s.}
\]
\label{NAloc}
\end{definition}

\section{Main results}
\label{mainresult}

We are now able to state the paper\rq{}s main results, which proof\rq{}s are given in Section~\ref{proof_section}.

\begin{theorem}[Characterization of $NA(\mathcal{Q}^T)$]\label{main_result}
    Assume the (PD) axiom. The following conditions are equivalent under Assumptions~\ref{assumption_prices} and~\ref{assumption:graph_Q}.
    \begin{enumerate}
        \item[(i)]  $NA(\mathcal{Q}^T)$ holds true.
        \item[(ii)]  There exists $P^*\in \mathcal{Q}^T$ such that $\mathrm{Aff}(D_{P^*}^{t+1})(\cdot)=\mathrm{Aff}(D^{t+1})(\cdot) \ \mathcal{Q}^t\mbox{-q.s.}$  and $0 \in \mathrm{Ri}(\mathrm{Conv}(D_{P^*}^{t+1}))(\cdot) \ \mathcal{Q}^t\mbox{-q.s.}$ for all $t\in \{0,\dots,T-1\}$.
    \end{enumerate}
\end{theorem}
We denote by $\mathcal{H}_T$ the set containing all the probability measures $P^*$ as in (ii). So, Theorem~\ref{main_result} says that $NA(\mathcal{Q}^T)$ is equivalent to $\mathcal{H}_T \neq \emptyset$. 
Theorem~\ref{main_result} was proved by Blanchard and Carassus in the setup of Bouchard and Nutz (see \cite[Theorem 3.29]{blanchard2020}) and has been conjectured by Carassus and Ferhoune in the projective setup. The direction (ii) implies (i) has been proved there, see \cite[Lemma 1 (iv) and  Remark 3]{carassus2024}. Note that, for all $t\in \{0,\dots,T-1\}$, the $\mathcal{Q}^t$-full-measure set where   (ii) holds true is the set $\Omega_{NA}^t$ introduced in Proposition \ref{equivalence_local_global} below. 
\\

The next proposition gives the characterization of the single-prior no-arbitrage condition $NA(P)$, and so, generalizes \cite[Theorem 3]{JS98} to the projective setup. It is a direct consequence of Theorem \ref{main_result} applied to $\mathcal{Q}^{T}:=\{P\},$  with $P:=p_{1} \otimes p_{2} \otimes \dots \otimes p_{T}.$ 
This is not the case in the Bouchard and Nutz setting since  $\mbox{Graph}(p_{t})$ belongs a priori to $\mathcal{B}_{c}(\O^{t}) \otimes \mathcal{B}(\mathfrak{P}(\O_{t+1}))$ and not to the analytic sets of $\O^{t} \times  \mathfrak{P}(\O_{t+1})$.
\begin{proposition}[Characterization of $NA(P)$]
\label{singleP}
Assume the (PD) axiom.  Assume that  Assumption~\ref{assumption_prices} holds true  and let $P \in   \mathfrak{P}(\O^{T})$ with the fixed disintegration $P:=p_{1} \otimes p_{2} \otimes \dots \otimes p_{T}$ where  $p_{t} \in  SK_{t}$ for all  $t\in \{1,\ldots,T\}$. Then, the $NA(P)$ condition holds if and only if  $0 \in \mbox{Ri}(\mbox{Conv}(D_{P}^{t+1}))(\cdot)$ $P^{t}$-a.s. for all $0 \leq t \leq T-1$.
\end{proposition}

The last theorem generalizes \cite[Theorem 3.6]{blanchard2020} to the projective setup and is an easy consequence of Theorem \ref{main_result} and Proposition \ref{singleP}. It proposes a meaningfull caracterization of $NA(\mathcal{Q}^{T})$ by the existence of a subclass of priors $\mathcal{P}^{T} \subseteq \mathcal{Q}^{T}$ such that $\mathcal{P}^{T}$ and $ \mathcal{Q}^{T}$ have the same polar sets and $NA(P)$ holds for all $P \in \mathcal{P}^{T}$. Note that this is not the case for $\mathcal{Q}^{T}$, where some priors may lead to arbitrage. 
So, instead of $NA(\mathcal{Q}^{T})$, one may assume that every model in $\mathcal{P}^{T}$ is arbitrage-free. Under quasi-sure uncertainty, this characterization allows tractable theorems for the existence of solutions to the problem of robust utility maximisation (see  \cite{blanchard2020},  \cite{BCK19} or \cite{RaMe18}). 

\begin{theorem}[Characterization of $NA(\mathcal{Q}^T)$]
\label{TheoS}
Assume the (PD) axiom. The following conditions are equivalent under Assumptions~\ref{assumption_prices} and~\ref{assumption:graph_Q}. \begin{itemize}
 \item[(i)] $NA(\mathcal{Q}^{T})$ holds true.
 \item[(ii)]  There exists some $\mathcal{P}^{T} \subseteq \mathcal{Q}^{T}$  such that $\mathcal{P}^{T}$ and  $\mathcal{Q}^{T}$ have  the  same polar sets  and such that   $NA(P)$ holds  for all $P \in \mathcal{P}^{T}.$
\end{itemize}
\end{theorem}

To prove Theorem~\ref{main_result}, we start by working in a one-period framework (see Proposition~\ref{prop_construction_P_star} in the Appendix). Then, we generalize the result to the multi-period framework using measurable selection techniques to find stochastic kernels.  
However, it requires first proving that the quasi-sure no-arbitrage is consistent with the local no-arbitrage at each time step. This is Proposition~\ref {equivalence_local_global}, which generalizes \cite[Theorem 4.5]{bouchard2015} in the projective framework under the (PD) axiom.

\begin{proposition}[Equivalence between global and local no-arbitrage]\label{equivalence_local_global}
Assume the (PD) axiom. The following conditions are equivalent under Assumptions~\ref{assumption_prices} and~\ref{assumption:graph_Q}.
\begin{enumerate}
    \item[(i)] $NA(\mathcal{Q}^T)$ holds true.
    \item[(ii)] For all $t \in \{0, \dots, T-1\}$, there exists a projective set  $\Omega_{NA}^t$ of $\mathcal{Q}^t$-full-measure, such that $NA(\mathcal{Q}_{t+1}(\omega^t))$ holds true for all $\omega^t \in \Omega_{NA}^t$.
\end{enumerate}
\end{proposition}

\section{Proofs of the main results}\label{proof_section}
\label{proofs}
\subsection{Proof of Proposition~\ref{equivalence_local_global}}
We first show Proposition~\ref{equivalence_local_global}, which allows us to prove Theorem~\ref {main_result} from the one-period result given in Proposition~\ref{prop_construction_P_star} in the Appendix. The proof that (ii) implies (i) is based on the same kind of ideas as \cite[Lemma 4.6]{bouchard2015}, using Proposition \ref{prop_conseq_PD_axiom} instead of Jankov-von Neumann\rq{}s measurable selection theorem. But the proof that (i) implies (ii) differs completely from that of \cite[Theorem 4.5]{bouchard2015}. We prove the claim by induction on $T,$ as it can be done in the single-prior case, and this is possible thanks to Corollary~\ref{corollary_section_B_omega} in the appendix, which allows us to transform $\mathcal{Q}^{t+1}\mbox{-q.s.}$ inequality to $\mathcal{Q}_{t+1}(\omega^{t})\mbox{-q.s.}$ one for all $\o^t$ in a projective set of $\mathcal{Q}^t$-full-measure. To prove the same claim, Bouchard and Nutz rely on a third condition: the existence of a certain type of martingale measure. Their proof is therefore based on a measurable selection argument for martingale measures.\\

\noindent \textit{(ii) implies (i).} \\
Assume that (ii) holds. We prove inductively on $T$ that (ii) holds. \\
If $T=1$, $NA(\mathcal{Q}^{1})=NA(\mathcal{Q}_1(\omega^0))$ holds true as $\Omega^0=\{\omega^0\}$ and $\mathcal{Q}^1 = \mathcal{Q}_1(\omega^0)$.\\
We  fix $t\in \{1,\dots,T-1\}$ and assume the claim at time $t$, i.e. that if for all $s \in \{0, \dots, t-1\}$, there exists a projective set $\Omega_{NA}^s$ of $\mathcal{Q}^s$-full-measure, such that for all $\omega^s \in \Omega_{NA}^s, \; NA(\mathcal{Q}_{s+1}(\omega^s))$ holds true, then $NA(\mathcal{Q}^t)$ holds true.\\
Now, suppose that for all $s \in \{0, \dots, t\}$, there exists a projective set $\Omega_{NA}^s$ of $\mathcal{Q}^s$-full-measure, such that for all $\omega^s \in \Omega_{NA}^s, \; NA(\mathcal{Q}_{s+1}(\omega^s))$ holds true. We prove that $NA(\mathcal{Q}^{t+1})$ holds true. Note first that $NA(\mathcal{Q}^{t})$ holds by induction. Let $\phi \in \Phi$ such that $V_{t+1}^{0, \phi} \geq 0 \; \mathcal{Q}^{t+1}\mbox{-q.s}$. Lemma~\ref{lemma_portfolio_measurability} shows that $V_{t+1}^{0, \phi}(\cdot)$ is $\mathbf{P}(\Omega^{t+1})$-measurable under Assumption~\ref{assumption_prices}. Then Corollary~\ref{corollary_section_B_omega}, under the (PD) axiom and Assumption~\ref{assumption:graph_Q}, implies that there exists $\bar{\Omega}^t\subseteq \Omega^{t}$, a projective set of $\mathcal{Q}^{t}$-full-measure, such that for all $\omega^{t} \in \bar{\Omega}^t$ 
 \begin{equation}\label{inequality_V}
    V_{t+1}^{0, \phi}(\omega^{t},\cdot) \geq 0 \;\mathcal{Q}_{t+1}(\omega^{t})\mbox{-q.s}.    
    \end{equation}
Let $\tilde{\Omega}^t := \bar{\Omega}^t\cap\Omega_{NA}^{t}$. Then, as $\tilde{\Omega}^t$ is the intersection of two projective sets, $\tilde{\Omega}^t$ is projective (see Proposition~\ref{prop_properties_projective_functions} (ii)). Moreover, $\tilde{\Omega}^t$ is also of $\mathcal{Q}^{t}$-full-measure as an intersection of full-measure sets.
    Let $\omega^{t}\in \tilde{\Omega}^t$. The previous arguments show that 
    \begin{equation}\label{inequation_proposition_equivalence_global_local}
      \phi_{t+1}(\omega^t)\Delta S_{t+1}(\omega^{t},\cdot) \geq -V_{t}^{0, \phi}(\omega^{t})\; \mathcal{Q}_{t+1}(\omega^{t})\mbox{-q.s.}  
    \end{equation}
 Assume for a moment that $\{V_t^{0,\phi}\geq 0\}$ is of $\mathcal{Q}^t$-full-measure. Then, as  $NA(\mathcal{Q}^t)$ holds, we get that $V_t^{0,\phi}= 0 \; \mathcal{Q}^{t}\mbox{-q.s.}$ Considering~(\ref{inequation_proposition_equivalence_global_local}) for $\omega^t$ in the intersection of $\tilde{\Omega}^t$ and $\{V_t^{0,\phi}= 0\}$, which is a projective set of $\mathcal{Q}^{t}$-full-measure (see Lemma~\ref{lemma_portfolio_measurability} again), we get that $$\phi_{t+1}(\omega^t)\Delta S_{t+1}(\omega^{t},\cdot) \geq 0\; \mathcal{Q}_{t+1}(\omega^{t})\mbox{-q.s.}$$
So, we can apply the local no-arbitrage $NA(\mathcal{Q}_{t+1}(\omega^t))$ to get that $\phi_{t+1}(\omega^t)\Delta S_{t+1}(\omega^{t},\cdot) = 0\; \mathcal{Q}_{t+1}(\omega^{t})\mbox{-q.s.}$ Therefore, using Fubini's theorem (recall that we are on a projective and full-measure set), it follows that $\phi_{t+1}\Delta S_{t+1}=0 \; \mathcal{Q}^{t+1}\mbox{-q.s.}$ and also $V_{t+1}^{0,\phi}=0 \; \mathcal{Q}^{t+1}\mbox{-q.s.}$, meaning that $NA(\mathcal{Q}^{t+1})$ holds as well. \\
It remains to prove that $V_t^{0,\phi}\geq 0 \; \mathcal{Q}^{t}\mbox{-q.s.}$ We consider the function $\phi^*_{t+1}=\phi_{t+1} \mathbf{1}_{\{V_t^{0,\phi}< 0\}}$. We have that $\phi^*$ is $\mathbf{P}(\Omega^{t})$-measurable (see Proposition~\ref{prop_properties_projective_functions}). Let $\omega^t\in \tilde{\Omega}^t\subseteq \bar{\Omega}^t$, we have that $\mathcal{Q}_{t+1}(\omega^{t})\mbox{-q.s.}$ 
    \begin{align*}
    \phi^*_{t+1}\Delta S_{t+1}(\omega^t,\cdot) &\geq V_t^{0,\phi}(\omega^t)1_{\{V_t^{0,\phi}<0\}}(\omega^t)+\phi^*_{t+1}(\omega^t)\Delta S_{t+1}(\omega^t,\cdot) \\
    &=V_{t+1}^{0,\phi}(\omega^t,\cdot)1_{\{V_t^{0,\phi}<0\}}(\omega^t)\geq 0,
    \end{align*}
    where we have used (\ref{inequality_V}) for the last inequality. As $\omega^t\in \tilde{\Omega}^t\subseteq \Omega_{NA}^t,$ we can apply the local no-arbitrage $NA(\mathcal{Q}_{t+1}(\omega^t))$ and we get that $
    \phi^*_{t+1}\Delta S_{t+1}(\omega^t,\cdot) =0 \;\mathcal{Q}_{t+1}(\omega^{t})\mbox{-q.s.}
    $
So, for all $\omega^t\in \tilde{\Omega}^t$, $\mathcal{Q}_{t+1}(\omega^{t})\mbox{-q.s.}$
    \begin{align*}
        0 \leq V_{t+1}^{0,\phi}(\omega^t,\cdot)1_{\{V_t^{0,\phi}<0\}}(\omega^t) &= V_t^{0,\phi}(\omega^t)1_{\{V_t^{0,\phi}<0\}}(\omega^t)+\phi^*_{t+1}(\omega^t)\Delta S_{t+1}(\omega^t,\cdot)\\
        &= V_t^{0,\phi}(\omega^t)1_{\{V_t^{0,\phi}<0\}}(\omega^t) \leq 0.
    \end{align*}
    Thus, $ V_{t}^{0,\phi}(\omega^t)1_{\{V_t^{0,\phi}<0\}}(\omega^t)=0$ for all $\omega^t\in \tilde{\Omega}^t$ which is of $\mathcal{Q}^t$-full-measure, and $V_{t}^{0,\phi}\geq 0 \; \mathcal{Q}^t\mbox{-q.s.}$ follows.\\

\noindent \textit{(i) implies (ii).} \\
Suppose now that $NA(\mathcal{Q}^{T})$ holds true. Fix $t\in \{0,\dots,T-1\}$.  Let 
$$\Omega_{NA}^t:= \bigl\{\omega^t\in \Omega^t:NA(\mathcal{Q}_{t+1}\bigl(\omega^t)\bigr) \; \mbox{holds}\bigr\}.$$
\noindent \textit{Step1: $\Omega_{NA}^t$ is a projective set.}\\
First, we rewrite the set $N^t$ where the local no-arbitrage fails:
    \vspace{-1em}
    \begin{align*}
        N^t&:= \Omega^t \setminus  \Omega_{NA}^t= \bigl\{\omega^t\in \Omega^t:NA\bigl(\mathcal{Q}_{t+1}(\omega^t)\bigr) \; \mbox{fails}\bigr\}\\
        &= \bigl\{\omega^t\in \Omega^t: \exists y\in \mathbb{R}^d,\exists q\in\mathcal{Q}_{t+1}(\omega^t) \;\mathrm{s.t.} \inf_{p\in\mathcal{Q}_{t+1}(\omega^t)} p \bigl[ y\Delta S_{t+1}(\omega^t,\cdot)\geq 0 \bigr] = 1 \mbox{ and } q \bigl[ y\Delta S_{t+1}(\omega^t,\cdot)> 0\bigr] > 0 \bigr\} \\
        &=\mathrm{proj}_{\Omega^t}\bigl[\{(\omega^t,q,y)\in \Omega^t \times \mathfrak{P}(\Omega_{t+1}) \times \mathbb{R}^d: q\in\mathcal{Q}_{t+1}(\omega^t),\; \lambda_{\mathrm{inf}}(\omega^t,q,y)=1 \;\mathrm{and}\;\lambda(\omega^t,q,y)\in (0,1] \}\bigr]\\
        &=\mathrm{proj}_{\Omega^t}\Bigl[ \bigl(\mathrm{Graph(\mathcal{Q}_{t+1})\times \mathbb{R}^d}\bigr) \bigcap \{\lambda_{\mathrm{inf}}=1\} \bigcap  \{\lambda\in ( 0,1 ]\}\Bigr]= \mathrm{proj}_{\Omega^t}(A),
    \end{align*}
    where $A:=( \mathrm{Graph(\mathcal{Q}_{t+1})\times \mathbb{R}^d} ) \cap \{\lambda_{\mathrm{inf}}=1\} \cap  \{\lambda\in ( 0,1 ]\}$ and the functions $\lambda$ and $\lambda_{\mathrm{inf}}$ are defined as follows:
    \begin{align*}
        \lambda:
        &\begin{cases}
            \Omega^t \times \mathfrak{P}(\Omega_{t+1}) \times \mathbb{R}^d \to \mathbb{R}\\
            (\omega^t,q,y)\mapsto q[y \Delta S_{t+1}(\omega^t,\cdot)> 0]=\displaystyle \int_{-}\mathbf{1}_{\{y \Delta S_{t+1}(\omega^t,\omega_{t+1})>0\}}q(d\omega_{t+1})
        \end{cases}\\
        \lambda_{\mathrm{inf}}:
        &\begin{cases}
            \Omega^t \times \mathfrak{P}(\Omega_{t+1}) \times \mathbb{R}^d \to \mathbb{R}\\
            (\omega^t,q,y)\mapsto\inf_{p\in \mathcal{Q}_{t+1}(\omega^t)}p[y \Delta S_{t+1}(\omega^t,\cdot)\geq0].
        \end{cases}
    \end{align*}
We now prove that $A$ is a projective set. Using Assumption~\ref{assumption:graph_Q} and $\mathbb{R}^d \in \mathcal{B}(\mathbb{R}^d)\subseteq \mathbf{P}(\mathbb{R}^d)$, we get that $\mathrm{Graph}(\mathcal{Q}_{t+1})\times \mathbb{R}^d \in \mathbf{P}(\Omega^{t}\times\mathfrak{P}(\Omega_{t+1})\times \mathbb{R}^d)$ (see Proposition~\ref{prop_properties_projective_functions} (iii)). Assume for a moment that $\lambda$ and $\lambda_{\mathrm{inf}}$ are projective functions. Then, $\{\lambda_{\mathrm{inf}}=1\} $ and $ \{\lambda\in ( 0,1 ]\}$ belong to $\mathbf{P}(\Omega^t\times\mathfrak{P}(\Omega_{t+1})\times \mathbb{R}^d)$.  
 Now, Proposition~\ref{prop_properties_projective_functions} (ii) provides closeness under finite intersections, implying that $A\in\mathbf{P}(\Omega^t\times\mathfrak{P}(\Omega_{t+1})\times \mathbb{R}^d)$ as well, and also stability under projection and complement, resulting in $N^t\in \mathbf{P}(\Omega^t)$ and  $\Omega_{NA}^{t}=\Omega^t\setminus N^t \in \mathbf{P}(\Omega^t).$ \\
It remains to prove that $\lambda$ and $\lambda_{\mathrm{inf}}$ are projective.   Let $J=(0,+\infty)$ or $J=[0,+\infty).$ 
Applying Proposition \ref{prj_cvt_pj}  to the stochastic kernel $p$ defined by $p(d\omega_{t+1}|(\omega^t,y,q))=q(d\omega_{t+1})$, which is 
Borel and thus projectively measurable and to 
$f(\omega^t,y,q,\o_{t+1})=\mathbf{1}_{\{y \Delta S_{t+1}(\omega^t,\omega_{t+1}) \in J\}},$  which is  projective (see Assumption \ref{assumption:graph_Q} and Proposition \ref{prop_properties_projective_functions}), we obtain that  
$$\alpha^J: \Omega^t \times \mathbb{R}^d\times \mathfrak{P}(\Omega_{t+1}) \ni (\omega^t,y,q) \mapsto \displaystyle \int_{-}\mathbf{1}_{\{y \Delta S_{t+1}(\omega^t,\omega_{t+1}) \in J\}}q(d\omega_{t+1})$$ 
is a projective function.  Using \cite[Proposition 8]{carassus2024} with 
$D=\{(\omega^t,y,q)\in \Omega^t \times \mathbb{R}^d\times \mathfrak{P}(\Omega_{t+1}): q \in \mathcal{Q}_{t+1}(\o^t)\}$ we get that 
$$\alpha^J_{\mathrm{inf}}: \Omega^t \times \mathbb{R}^d  \ni (\omega^t,y)\mapsto \inf_{q\in\mathcal{Q}_{t+1}(\omega^t)} \alpha^J(\omega^t,y,q)$$ is also projective. 
 Then,  as measurability is preserved by composition with Borel (thus projective) functions (see Proposition \ref{prop_properties_projective_functions} (vi)), we conclude by remarking that $\lambda=\alpha^{(0,+\infty)} \circ \iota$ where $\iota:\Omega^t \times \mathfrak{P}(\Omega_{t+1}) \times \mathbb{R}^d \ni (\omega^t,q,y) \mapsto (\omega^t,y,q)$ is Borel and $\lambda_{\mathrm{inf}}=\alpha^{[0,+\infty)}_{\mathrm{inf}} \circ \rho$ where $\rho:  \Omega^t \times \mathfrak{P}(\Omega_{t+1}) \times \mathbb{R}^d \ni (\omega^t,q,y)\in \mapsto (\omega^t,y)$ is also Borel. 


\noindent \textit{Step2: $\Omega_{NA}^t$ is a $\mathcal{Q}^t$-full-measure set.}\\    
Suppose by contraposition that $N^t$ is not $\mathcal{Q}^t$-polar. This means that $P^t[N^t]>0$ for some $P\in \mathcal{Q}^T$ having the disintegration $P:=p_1\otimes\dots\otimes p_T$. We apply now measurable selection to find an intertemporal strategy of arbitrage contradicting the quasi-sure no-arbitrage hypothesis.
    
    As $A\in\mathbf{P}(\Omega^t\times\mathfrak{P}(\Omega_{t+1})\times \mathbb{R}^d)$, Proposition~\ref{prop_conseq_PD_axiom} gives the existence of a function $\Xi=(q^*,\phi^*):\mathrm{proj}_{\Omega^t}(A)=N^t\to \mathfrak{P}(\Omega_{t+1})\times \mathbb{R}^d$, $\mathbf{P}(\Omega^t)$-measurable and such that $\mathrm{Graph}(\Xi)\subseteq A$.  
So, for all $\omega^t\in N^t,$ $q^*(\o^t)$ is a probability measure on $\Omega_{t+1}$ and we  write $q^*(\omega^t)=q^*(\cdot \mid \omega^t).$ 
Moreover, $ N^t \ni \omega^t  \mapsto  q^*(\cdot \mid \omega^t) \in \mathfrak{P}(\Omega_{t+1})$  is $\mathbf{P}(\Omega^t)$-measurable, 
and the inclusion $\mathrm{Graph}(\Xi)\subseteq A$ implies that for all $\omega^t\in N^t$, $q^*(\cdot \mid \omega^t)\in \mathcal{Q}_{t+1}(\omega^t).$ 
We also have that $\phi^*$ is $\mathbf{P}(\Omega^t)$-measurable and for all $\omega^t\in N^t$
    \begin{equation}\label{eq_inf_proof_Omega_NA}
     \inf_{p\in \mathcal{Q}_{t+1}(\omega^t)}p[\phi^*(\omega^t) \Delta S_{t+1}(\omega^t,\cdot)\geq 0]=1 \; \mathrm{and} \; q^*\bigl(\phi^*(\omega^t)\Delta S_{t+1}(\omega^t,\cdot) >0\mid \omega^t \bigr)>0.
    \end{equation}
    
    We set $\hat{\phi}_{t+1} := \phi^* $ on $N^t$, $\hat{\phi}_{t+1} := 0 $ on $\Omega^t\setminus N^t$, and $\hat{\phi}_s:=0$ for $s\neq t+1$. We also set $\hat{q}:=q^*$ on $N^t$, $\hat{q}:=\tilde{q}$ on $\Omega^t\setminus N^t$, where $\tilde{q} \in SK_{t+1}$ is such that $\tilde{q}(\cdot\mid\omega^t) \in \mathcal{Q}_{t+1}(\omega^t)$ for all $\omega^t\in \Omega^t$ ($\tilde{q}$ is obtained by performing measurable selection on $\mathrm{Graph}(\mathcal{Q}_{t+1})$ as Assumption~\ref{assumption:graph_Q} holds). This defines a strategy and a stochastic kernel, which are indeed projectively measurable (see Proposition~\ref{prop_properties_projective_functions} and the proof of Lemma~\ref{lemma_section_B_omega} where similar results are proved with more details). We now show that $\hat{\phi}$ is an arbitrage.
    
    Let $s\in \{1,\dots,T\}$. By construction of $\hat{\phi}\in \Phi$, for all $\omega^{s-1}\in \Omega^{s-1}$, $\hat{\phi}_s(\omega^{s-1}) \Delta S_s(\omega^{s-1},\cdot) \geq 0 \; \mathcal{Q}_s(\omega^{s-1})\mbox{-q.s.}$ (see (\ref{eq_inf_proof_Omega_NA})). Fubini's theorem is then applied to obtain that $\hat{\phi}_s \Delta S_s \geq 0 \; \mathcal{Q}^T\mbox{-q.s}$. We conclude that $$\displaystyle \sum_{s=1}^T\hat{\phi}_s \Delta S_s \geq 0 \; \mathcal{Q}^T\mbox{-q.s}.$$
    
    Moreover, we define $\hat{P} := P^t \otimes \hat{q} \otimes p_{t+2} \otimes \dots \otimes p_T.$    
    Then $\hat{P} \in \mathcal{Q}^T$ by construction, and using Fubini's theorem:
    \begin{align*}        
        \hat{P} \Bigl[ \sum_{s=1}^T \hat{\phi}_s \Delta S_s > 0 \Bigr] &= 
        \int_{\Omega^T} \mathbf{1}_{\{ \sum_{s=1}^T \hat{\phi}_s \Delta S_s > 0\}}(\omega^{T})
        \, \hat{P}(d\omega^{T}) \notag \\
        &= \int_{\Omega^{t+1}}\mathbf{1}_{\{ \hat{\phi}_{t+1} \Delta S_{t+1} > 0 \}}(\omega^{t+1}) 
        \, P^t\otimes \hat{q}(d\omega^{t+1}) \notag \\
        &= \int_{\Omega^t}   \int_{\Omega_{t+1}} \mathbf{1}_{\{ \hat{\phi}_{t+1} \Delta S_{t+1} > 0 \}}(\omega^t,\omega_{t+1}) 
        \,  \hat{q}(d\omega_{t+1}\mid \omega^{t})P^t(d\omega^{t}) \notag \\
        &=\int_{N^t}  q^*(\phi^*(\omega^t)\Delta S_{t+1}(\omega^t,\cdot) >0\mid \omega^t)P^t(d\omega^{t}) > 0,
    \end{align*}
    as the integral of a strictly positive function (see (\ref{eq_inf_proof_Omega_NA})) on a non-null set (relative to the measure $P^t$). 
    So, $\hat{\phi}$ is an intertemporal arbitrage, which contradicts $NA(\mathcal{Q}^T)$.

\subsection{Proof of Theorem ~\ref{main_result}.}

\textit{Reverse implication.} 

This is proved in \cite[Lemma 1 (iv)]{carassus2024}.\\

\textit{Direct implication.}\\
The proof is an adaptation of \cite[Theorem 3.29]{blanchard2020}'s one to the projective setup. 
Assume that $NA(\mathcal{Q}^T)$ holds true. Proposition~\ref{equivalence_local_global} shows that for all $t \in \{0, \dots, T-1\}$, there exists a projective set $\Omega_{NA}^t$ of $\mathcal{Q}^t$-full-measure, such that $NA(\mathcal{Q}_{t+1}(\omega^t))$ holds true for all $\omega^t \in \Omega_{NA}^t$.

\textit{Construction of $\mathcal{E}_{t+1}$ and measurable selection.}\\
Fix  $t \in \{0, \dots, T-1\}$. Let $\mathcal{E}_{t+1} : \Omega^t \twoheadrightarrow \mathfrak{P}(\Omega_{t+1})$ be defined 
for all $\omega^t \in \Omega^t$  by 
\[
\mathcal{E}_{t+1}(\omega^t) := \bigl\{ p \in \mathcal{Q}_{t+1}(\omega^t): \, 0 \in \text{Ri} \bigl( \text{Conv}(E^{t+1}) (\omega^t, p) \bigr) \mbox{ and } \text{Aff} ( E^{t+1})(\omega^t, p) = \text{Aff} \bigl( D^{t+1} \bigr)(\omega^t) \bigr\}.
\]
Let $\o^t \in \O^t$. Recalling Definitions \ref{NAloc} and \ref{NAone} and applying Proposition~\ref{prop_construction_P_star} in the Appendix, we get that 
\[
\begin{aligned}
\text{NA}\bigl(\mathcal{Q}_{t+1}(\omega^t)\bigr) \text{ holds true}
&\;\Longrightarrow\;
\exists\,p \in \mathcal{Q}_{t+1}(\omega^t) \;\text{with}\;
0 \in \mathrm{Ri}\Bigl(\mathrm{Conv}\,\bigl(E^{t+1}(\omega^t,p)\bigr)\Bigr) 
\\&\phantom{\;\Longleftrightarrow\;}\;\;
\text{and}\;\;
\mathrm{Aff}\bigl(E^{t+1}(\omega^t,p)\bigr)
=
\mathrm{Aff}\bigl(D^{t+1}(\omega^t)\bigr) \\
&\;\Longleftrightarrow\;
\mathcal{E}_{t+1}(\omega^t) \neq \emptyset.
\end{aligned}
\]
Thus, we deduce that $\Omega^t_{NA} \subseteq \{ \mathcal{E}_{t+1} \neq \emptyset \} $. Suppose for a moment that $\operatorname{Graph} (\mathcal{E}_{t+1} )$ is a projective set. Using Proposition~\ref{prop_conseq_PD_axiom} for $\operatorname{Graph} (\mathcal{E}_{t+1} )$  gives the existence  for all  $\o^t \in  \Omega^t_{NA}$ of 
$\hat p_{t+1}(\cdot|\o^t)  \in  \mathfrak{P}(\Omega_{t+1})$   such that 
$\Omega^t_{NA} \ni \o^t  \mapsto \hat p_{t+1}(\cdot|\o^t)\in  \mathfrak{P}(\Omega_{t+1})$ is projectively measurable  and  $\hat p_{t+1}(\cdot| \omega^t) \in \mathcal{E}_{t+1}(\omega^t) $ for every $\omega^t \in \Omega^t_{NA}$. Indeed, we have that $\text{proj}_{\O^t}\text{Graph}(\mathcal{E}_{t+1})= \{ \mathcal{E}_{t+1} \neq \emptyset \} \supseteq \Omega^t_{NA}. $ Let $\tilde{q}_{t+1}\in SK_{t+1} $ be obtained by performing measurable selection on $\mathrm{Graph}(\mathcal{Q}_{t+1})$ as Assumption~\ref{assumption:graph_Q} holds. We set ${q}^*_{t+1}:=\hat p_{t+1}$ on $\Omega^t_{NA}$ and  ${q}^*_{t+1}:=\tilde{q}_{t+1}$ on $\Omega^t\setminus\Omega^t_{NA}$. 
Define $P^* := p^*_1 \otimes \cdots \otimes p^*_T $. By construction of $P^*,$ as $\Omega^t_{NA}$  and $\Omega^t\setminus\Omega^t_{NA}$ are projective sets, we have that $P^* \in Q^T $. Furthermore, using Remark~\ref{EP_included_D} and  ${q}^*_{t+1}:=\hat p_{t+1}$ on $\Omega^t_{NA}$, we obtain for all $\omega^t \in \Omega^t_{NA}$ that 
\begin{eqnarray*}
\text{Aff} \bigl( D^{t+1}_{P^*} \bigr) (\omega^t) = \text{Aff} \bigl( E^{t+1} \bigr)\bigl(\omega^t, p^*_{t+1}(\cdot| \omega^t)\bigr) = \text{Aff} \bigl( E^{t+1} \bigr)\bigl(\omega^t, \hat{p}_{t+1}(\cdot| \omega^t)\bigr)=\text{Aff} \bigl( D^{t+1} \bigr)(\omega^t)
\\
0 \in \text{Ri} \Bigl( \text{Conv}(E^{t+1}) \bigl(\omega^t, \hat{p}_{t+1}(\cdot| \omega^t)\bigr) \Bigr) =\text{Ri} \Bigl( \text{Conv}(E^{t+1}) \bigl(\omega^t, p^*_{t+1}(\cdot| \omega^t)\bigr) \Bigr) = \text{Ri} \bigl( \text{Conv}(D^{t+1}_{P^*}) \bigr)(\omega^t),
\end{eqnarray*}
and this will conclude the proof as $\Omega_{NA}^t$ is a $\mathcal{Q}^t$-full-measure set.  

\textit{Proof of $\operatorname{Graph} (\mathcal{E}_{t+1} )\in 
\mathbf{P}(\Omega^t \times \mathfrak{P}(\Omega_{t+1})).$}\\
Let
\begin{eqnarray*}
B &  :=  &\bigl\{ (\omega^t, p)\in \Omega^t\times\mathfrak{P}(\Omega_{t+1}): \ \operatorname{Ri} \bigl(\overline{\operatorname{Conv}}(E^{t+1}) \bigr)(\omega^t, p) \cap \{ 0 \} \neq \emptyset \bigr\}\\ 
C & := & \bigl\{ (\omega^t, p)\in \Omega^t\times\mathfrak{P}(\Omega_{t+1}): \ \text{Aff} \bigl( E^{t+1} \bigr) (\omega^t, p) = \text{Aff} \bigl( D^{t+1} \bigr)(\omega^t) \bigr\}.
\end{eqnarray*}
Recall from Proposition~\ref{lemma_measurability_sets} (i) et (iii) that 
$
 \text{Ri} (\overline{\operatorname{Conv}}(E^{t+1}))
$
 is closed-valued and $\Delta^1_n (\Omega^t \times \mathfrak{P}(\Omega_{t+1})) $-measurable for some $n\geq 1$ and  that $\text{Aff} (D^{t+1} )$ 
is $\Delta^1_q(\Omega^t ) $-measurable for some $q\geq 1$. We assume without loss of generality that $q\leq n$, changing $n$ into  $\max(n,q)$ if necessary. So, using 
Proposition~\ref{prop_properties_projective_functions} (i) $\text{Aff} (D^{t+1} )$ 
is also $\Delta^1_n(\Omega^t ) $-measurable. \\
We apply \cite[Theorem 14.3]{RockafellarWets1998} in the measurable space $(\Omega^t \times \mathfrak{P}(\Omega_{t+1}), \Delta^1_n(\Omega^t \times \mathfrak{P}(\Omega_{t+1}))) $ and   we conclude that  $B  \in \Delta^1_n(\Omega^t \times \mathfrak{P}(\Omega_{t+1}))$. 
It also implies that $B \in \mathbf{P}(\Omega^t \times \mathfrak{P}(\Omega_{t+1})) $.\\
Let $h : \Omega^t \times \mathfrak{P}(\Omega_{t+1}) \to \mathbb{R} $ be defined by
\begin{eqnarray}
\label{eq:distance_definition}
h(\omega^t, p) := d \bigl( \operatorname{Aff} \bigl( E^{t+1}(\omega^t, p) \bigr), \operatorname{Aff} \bigl( D^{t+1}(\omega^t) \bigr) \bigr)
= \sup_{x \in \mathbb{R}^d} \Big| d \bigl( x, \operatorname{Aff} \bigl( E^{t+1}(\omega^t, p) \bigr) \bigr) 
- d \bigl( x, \operatorname{Aff} \bigl( D^{t+1}(\omega^t) \bigr) \bigr) \Big|. 
\end{eqnarray}
Here $d(F, G) $ is the Hausdorff distance between two nonempty sets $F, G \subseteq \mathbb{R}^d $, see for instance \cite[Definition 3.70 and Lemma 3.74]{AliprantisBorder2006} and $d(x, F) = \inf \{ |x-y|: \, y \in F \}.$ 
Proposition~\ref{lemma_measurability_sets} (i) also shows that $\text{Aff} ( E^{t+1})$ 
is $\Delta^1_n(\Omega^t \times \mathfrak{P}(\Omega_{t+1})) $-measurable and 
applying  \cite[Theorem 18.5]{AliprantisBorder2006} with the same measurable space as before, 
we  conclude that 
$$
\Omega^t \times \mathfrak{P}(\Omega_{t+1}) \times \mathbb{R}^d \ni \bigl((\omega^t, p), x\bigr)  \mapsto d\Bigl(x, \operatorname{Aff}\bigl(E^{t+1}(\omega^t, p)\bigr)\Bigr)
$$
is a Caratheodory function. This means that for every $x \in \mathbb{R}^d $, 
$
\Omega^t \times \mathfrak{P}(\Omega_{t+1}) \ni (\omega^t, p)  \mapsto d(x, \operatorname{Aff}(E^{t+1}(\omega^t, p)))
$
is $\Delta^1_n(\Omega^t \times \mathfrak{P}(\Omega_{t+1}))$-measurable and 
for every $(\omega^t, p) \in \Omega^t \times \mathfrak{P}(\Omega_{t+1})$,  
$
\mathbb{R}^d \ni x \ \mapsto d(x, \operatorname{Aff}(E^{t+1}(\omega^t, p)))
$
is continuous. 
Now, we apply again
  \cite[Theorem 18.5]{AliprantisBorder2006} to $\text{Aff} (D^{t+1} )$ with the measurable space 
  $(\Omega^t, \Delta^1_n(\Omega^t ))$ and  
we get that 
$$
\Omega^t  \times \mathbb{R}^d  \ni (\omega^t, x) \mapsto d\Bigl(x, \operatorname{Aff}\bigl(D^{t+1}(\omega^t)\Bigl)\Bigl)
$$
is a Caratheodory function, which implies that  for every 
$x \in \mathbb{R}^d $, 
$
\Omega^t  \ni \omega^t  \mapsto d(x, \operatorname{Aff}(D^{t+1}(\omega^t)))
$
is $\Delta^1_n(\Omega^t )$-measurable  and 
for every $\omega^t\in \Omega^t $,  the function 
$
\mathbb{R}^d  \ni x \mapsto d(x, \operatorname{Aff}(D^{t+1}(\omega^t)))
$
is continuous. So, $
\mathbb{R}^d  \ni x  \mapsto | d \bigl( x, \operatorname{Aff} ( E^{t+1}(\omega^t, p) ) ) 
- d ( x, \operatorname{Aff} ( D^{t+1}(\omega^t) ) ) |
$ is continuous and 
we can replace $\mathbb{R}^d $ with $ \mathbb{Q}^d $ in \eqref{eq:distance_definition}. 
Then, Proposition~\ref{prop_properties_projective_functions} (v) and (viii)  shows that 
$(\omega^t, p)
\mapsto | d \bigl( x, \operatorname{Aff} ( E^{t+1}(\omega^t, p) ) ) 
- d ( x, \operatorname{Aff} ( D^{t+1}(\omega^t) ) ) |
$
is $\Delta^1_n(\Omega^t \times \mathfrak{P}(\Omega_{t+1})) $-measurable and that  $h $ is also $\Delta^1_n(\Omega^t \times \mathfrak{P}(\Omega_{t+1})) $-measurable, as a countable supremum. So, we obtain that 
$$C = h^{-1}(\{0\})  \in \Delta^1_n\bigl(\Omega^t \times \mathfrak{P}(\Omega_{t+1})\bigr)  \subseteq \mathbf{P}\bigl(\Omega^t \times \mathfrak{P}(\Omega_{t+1})\bigr).$$ 
As $\operatorname{Ri} (\overline{\operatorname{Conv}}(E^{t+1}) )=\operatorname{Ri} ({\operatorname{Conv}}(E^{t+1}) )$, see \cite[Theorem 6.3]{rockafellar1970convex}, 
Assumption~\ref{assumption:graph_Q} and Proposition~\ref{prop_properties_projective_functions} 
 (ii) show that 
\[
\operatorname{Graph} (\mathcal{E}_{t+1}) = \operatorname{Graph} (\mathcal{Q}_{t+1}) \cap B \cap C \in 
\mathbf{P}(\Omega^t \times \mathfrak{P}\bigl(\Omega_{t+1})\bigr).
\]
Now, the proof is complete.
\begin{remark}
In this proof, it is important at several stages to work with $\Delta^1_n$ and not $\mathbf{P}$. Indeed, we use that $\Delta^1_n$ is a sigma-algebra to apply \cite[Theorem 18.5]{AliprantisBorder2006} and get Caratheodory functions. We also need that the class of $\Delta^1_n$-measurable functions is closed by countable supremum. These two properties are not true for $\mathbf{P}$. 
\end{remark}

\subsection{Proof of Proposition~\ref{singleP}.}
Let $P \in   \mathfrak{P}(\O^{T})$ with the fixed disintegration $P:=p_{1} \otimes p_{2} \otimes \dots \otimes p_{T}$ where  $p_{t} \in  SK_{t}$ for all  $t\in \{1,\ldots,T\}$. 
We want to apply Theorem \ref{main_result} to $\mathcal{Q}^{T}:=\{p_{1} \otimes p_{2} \otimes \dots \otimes p_{T}\}$. 
For that we need to prove that  $\mbox{Graph}(p_{t+1}) \in \mathbf{P}(\Omega^t \times \mathfrak{P}(\Omega_{t+1}))$ for all  $t\in \{0,\ldots,T-1\}$. Remark that 
$$\mbox{Graph}(p_{t+1})= \big\{ (\omega^t , q) \in \O^t \times \mathfrak{P}(\Omega_{t+1}): \ q=p_{t+1}(\cdot| \omega^t)  \big\}.$$
Since $p_{t+1} \in SK_{t+1},$ we get that  $h:$ $\Omega^t \times \mathfrak{P}(\Omega_{t+1})  \ni (\omega^t, q) \mapsto p_{t+1}(\cdot| \omega^t) - q \in \mathfrak{P}(\Omega_{t+1})$ is $ \mathbf{P} (\Omega^t \times \mathfrak{P}(\Omega_{t+1})) $-measurable  (see Proposition~\ref{prop_properties_projective_functions} (iv) and (viii)) and 
$$\mbox{Graph}(p_{t+1})=h^{-1}(0)\in \Delta^1_{\ell }(\Omega^t \times \mathfrak{P}(\Omega_{t+1}) )\subseteq\mathbf{P} (\Omega^t \times \mathfrak{P}(\Omega_{t+1}))
$$ for some $\ell \geq 1$, see Definition \ref{def_projective_mes_functions}.  So, Theorem \ref{main_result} with $\mathcal{Q}^{T}:=\{p_{1} \otimes p_{2} \otimes \dots \otimes p_{T}\}$ 
asserts that $NA(P)$ is equivalent to  $0 \in \mathrm{Ri}(\mathrm{Conv}(D_{P}^{t+1}))(\cdot) \ P^t\mbox{-a.s.}$ for all $t\in \{0,\dots,T-1\}$. Indeed, here as $\mathcal{Q}^{T}$ is a singleton, $P^*=P$ and $D^{t+1}=D_{P}^{t+1}.$

\subsection{Proof of Theorem ~\ref{TheoS}.}
The proof is copypaste from \cite[Theorem 3.6]{blanchard2020} and is given for the reader\rq{}s convenience. \\
{\it Step 1: Reverse implication.}\\
Assume now that there exists some $\mathcal{P}^{T} \subseteq \mathcal{Q}^{T}$ such that $\mathcal{P}^{T}$ and  $\mathcal{Q}^{T}$ have the same polar sets and the $NA(P)$ condition holds  for all $P \in \mathcal{P}^{T}$. 
If  $NA(\mathcal{P}^{T})$  fails, there exist some $\phi  \in \Phi$ and $P \in \mathcal{P}^{T}$ such that $V_{T}^{0,\phi} \geq 0 \; {\mathcal{P}^{T}}\mbox{-q.s.}$ and $P(V_{T}^{0,\phi}  > 0)>0:$ $NA(P)$  also fails. So, $NA(\mathcal{P}^{T})$ holds and  
also $NA(\mathcal{Q}^{T})$ as $\mathcal{P}^{T}$ and  $\mathcal{Q}^{T}$ have the same polar sets.\\
{\it Step 2: Direct implication.}\\
Theorem \ref{main_result} implies that there exists  some  $P^*\in \mathcal{Q}^{T}$ with the disintegration  $P^*:=p_1^{*}\otimes p_{2}^{*}\otimes \cdots \otimes p_{T}^{*}$  such that   
$\mbox{Aff}\bigl(D_{P^*}^{t+1}\bigr)(\o^{t})=\mbox{Aff} (D^{t+1})(\o^{t})$ and $0 \in \mbox{Ri} ({\mbox{Conv}}(D_{P^*}^{t+1}))(\o^{t})$
  for all  $\o^t$ in some  $\mathcal{Q}^{t}$-full-measure set, namely  $\Omega^{t}_{NA},$ and all $0 \leq t \leq T-1$. 
Let
$\mathcal{P}^{T}$ be defined recursively: $\mathcal{P}^{1}:=\{\ell p_{1}^*+ (1-\ell)p:\; p \in \mathcal{Q}^{1},\; 0<\ell \leq 1  \}$ and for all $1 \leq t \leq T-1$
 \begin{align}
\label{PstarARB}
 \mathcal{P}^{t+1}&:=\Bigl\{
 P \otimes  (\ell p^*_{t+1}+ (1-\ell) q): \ 0<\ell \leq 1, \,P \in \mathcal{P}^{t}, \,q \in SK_{t+1}, \, q(\cdot|\o^{t}) \in \mathcal{Q}_{t+1}(\o^{t}) \, \forall  \o^{t} \in \O^{t}\;    \Bigr\}.
\end{align}
\indent {\it (i) $\mathcal{P}^{t} \subseteq \mathcal{Q}^{t}$  for all $t\in \{1,\ldots,T\}$}.\\
\noindent  This follows by induction from   the convexity of $\mathcal{Q}_{t+1}(\o^{t})$; see \eqref{PstarARB} and recall that $p^{*}_{t+1}(\cdot \mid \o^{t}) \in \mathcal{Q}_{t+1}(\o^{t})$.\\
 \indent  {\it (ii)  $\mathcal{Q}^{t}$ and $\mathcal{P}^{t}$ have the same polar-sets  for all $t\in \{1,\ldots,T\}$}.\\
Fix some $t\in \{1,\ldots,T\}$.   As  $\mathcal{P}^{t} \subseteq \mathcal{Q}^{t}$,  it is clear that  a $\mathcal{Q}^{t}$-polar set is also a  $\mathcal{P}^{t}$-polar set. The other direction follows from \eqref{Pnt_pj} below, which is proved in \cite[Lemma 14]{CF24} by induction on $t$. Let $Q^t:= q_1 \otimes \cdots \otimes q_t\in\mathcal{Q}^t, $ then there exist some $(R_k^t)_{0\leq k\leq t-1}\subset \textup{Conv}(\mathcal{Q}^t)$ such that  
\begin{eqnarray}
P^t := \frac1{2^t} \biggl(Q^t +  \sum_{k=0}^{t-1} \binom{t}{k}  R_k^t\biggr) \in \mathcal{P}^{t}. \label{Pnt_pj}
\end{eqnarray}
So, 
$Q^{t} \ll P^{t}$ and   a $\mathcal{P}^{t}$-polar set is also a  $\mathcal{Q}^{t}$-polar set. \\
 \indent {\it (iii) $NA(P)$ holds for all $P\in \mathcal{P}^{T}$}. \\
Fix  some $P:=p_{1} \otimes p_{2} \otimes \cdots \otimes p_{T}\in \mathcal{P}^{T} \subseteq  \mathcal{Q}^{T}$, some $0 \leq t \leq T-1$ and $\o^{t} \in \O_{NA}^{t}$.  
We  establish that    $0 \in  \mbox{Ri} ({\mbox{Conv}}(D_{P}^{t+1}))(\o^{t})$. Then, as   $P^{t}(\O^{t}_{NA})=1,$  Proposition~\ref{singleP} shows that $NA(P)$ holds true and (iii)  follows. 
Remark \ref{EP_included_D} and \eqref{PstarARB} ($p^*_{t+1}(\cdot|\o^{t}) \ll p_{t+1}(\cdot|\o^{t})$) imply  that
${D}_{P^{*}}^{t+1}(\o^{t}) \subseteq {D}_{P}^{t+1}(\o^{t}) \subseteq  {D}^{t+1}(\o^{t})$. Thus, we see that $0 \in {\mbox{Conv}}(D_{P^*}^{t+1})(\o^{t}) \subseteq {\mbox{Conv}}(D_{P}^{t+1})(\o^{t})$. We  have that
$$\mbox{Aff} \bigl(D^{t+1}\bigr)(\o^{t}) = \mbox{Aff}\bigl(D_{P^*}^{t+1}\bigr)(\o^{t}) \subseteq \mbox{Aff} \bigl(D_P^{t+1}\bigr)(\o^{t}) \subseteq \mbox{Aff} \bigl(D^{t+1}\bigr)(\o^{t}).$$
As  $0 \in \mbox{Ri} ({\mbox{Conv}}(D_{P^*}^{t+1}))(\o^{t})$, there exists some $\varepsilon>0$ such that
$$
 B(0,\varepsilon) \bigcap \mbox{Aff}\bigl(D_{P}^{t+1}\bigr)(\o^{t}) =B(0,\varepsilon) \bigcap \mbox{Aff}\bigl(D_{P^*}^{t+1}\bigr)(\o^{t})
   \subseteq  {\mbox{Conv}}(D_{P^*}^{t+1})(\o^{t}) \subseteq {\mbox{Conv}}(D_{P}^{t+1})(\o^{t}),
$$
which concludes the proof of $0 \in  \mbox{Ri} ({\mbox{Conv}}(D_{P}^{t+1}))(\o^{t})$. 

\section{Appendix}
In this appendix, we first state and prove the results related to one-period models, which were used in the proof of Theorem ~\ref{main_result}. Then, we recall properties from \cite{carassus2024} on projective sets and functions. The next section proves technical results about the measurability of the supports, which were used in the proof of Theorem ~\ref{main_result}. Finally, we prove an extension of some results of  \cite{Carassus25} to the projective setup, which allows us to transform $\mathcal{Q}^{t+1}\mbox{-q.s.}$ inequality to $\mathcal{Q}_{t+1}(\omega^{t})\mbox{-q.s.}$ one for all $\o^t$ in a projective set of $\mathcal{Q}^t$-full-measure. This was crucial in the proof of Proposition~\ref{equivalence_local_global}

\subsection{One-period model}\label{subsection_one_period}

We introduce the one-period model and construct a probability measure for which the single-prior no-arbitrage condition holds in a quasi-sure sense. Let ${\Omega}$ be a Polish space, $\mathfrak{P}(\Omega) $ the set of all probability measures defined on ${\cal B}(\Omega) $, and $\mathcal{Q} $ a nonempty convex subset of $\mathfrak{P}(\Omega) $. 
Let $Y $ be a ${\cal B}_c(\Omega)$-measurable $\mathbb{R}^d $-valued random variable. The following sets are the pendants in the one-period case of the ones introduced in Definition~\ref{def_conditional_supports}. 
Let $p \in \mathcal{Q} $,
\begin{align*}
    E(p) &:= \bigcap \{ A \subseteq \mathbb{R}^d: \text{closed}, \, p[Y(\cdot) \in A] = 1 \}, \\
    D &:= \bigcap \{ A \subseteq \mathbb{R}^d: \text{closed}, \, q[Y(\cdot) \in A] = 1, \, \forall q \in \mathcal{Q}\}. 
\end{align*}
We now state the pendant of the no-arbitrage conditions in the one-period framework.
\begin{definition}[Quasi-sure one-period no-arbitrage condition]
The condition $NA(\mathcal{Q})$ holds true if
\[
hY(\cdot) \geq 0 \; \mathcal{Q}\text{-q.s. for some } h\in \mathbb{R}^d \implies hY(\cdot) = 0 \; \mathcal{Q}\text{-q.s.}
\]
\label{NAone}
\end{definition}

\begin{definition}[Single-prior one-period no-arbitrage condition]
Let $p \in \mathcal{Q}$. The condition $NA(p)$ holds true if
\[
hY(\cdot) \geq 0 \; p\text{-a.s. for some } h\in \mathbb{R}^d \implies hY(\cdot) = 0 \; p\text{-a.s.}
\]
\end{definition}
Note that when $\mathcal{Q}=\{p\}$, $NA(\mathcal{Q})$ and $NA(p)$ coincide. 

%
%



The following lemma recalls well-known results about supports and no-arbitrage in a one-period framework, which can, for example, be found in \cite{blanchard2020}. They are recalled for the reader's convenience.
\begin{lemma}\label{lemma_affine_hull_linear}
(i) If $0 \notin  \mathrm{Ri}(\mathrm{Conv}(D)),$ there exists some $h^* \in \mathrm{Aff}(D)$, $h^* \neq 0$ such that $h^* y \geq 0$ for all $y \in D$. \\
(ii) For any $h \in \mathbb{R}^d \setminus \{0\}$,
\begin{align}
\label{pareil1}
hY(\cdot) \geq 0 \; \mbox{$\mathcal{Q}$-q.s.} & \Longleftrightarrow h y \geq 0,\;  \forall y \in D\\
\label{pareil2}
hY(\cdot) = 0 \; \mbox{$\mathcal{Q}$-q.s.} & \Longleftrightarrow h y = 0,\;  \forall y \in D.
\end{align}
(iii) Let $p\in \mathcal{Q}$. For any $h \in \mathbb{R}^d \setminus \{0\}$,
\begin{align}
\label{pareil3}
hY(\cdot) \geq 0 \; \mbox{$p$-a.s.} & \Longleftrightarrow h y \geq 0,\;  \forall y \in E(p)\\
\label{pareil4}
hY(\cdot) =0 \; \mbox{$p$-a.s.} & \Longleftrightarrow h y = 0,\;  \forall y \in E(p).
\end{align}
(iv) Assume that $NA(\mathcal{Q})$ holds. Then, $0 \in \mathrm{Ri}(\mathrm{Conv}(D)).$ \\
(v)  Let $p\in \mathcal{Q}$. Assume that $NA(p)$ holds. Then, $0 \in \mathrm{Ri}(\mathrm{Conv}(E(p))).$
\end{lemma}
\begin{proof}
Assertion (i) is a classical exercise relying on a separation argument in $\mathbb{R}^{d}$, see  \cite[Theorems 11.1, 11.3]{rockafellar1970convex}).  
For (ii) and (iii), we only show \eqref{pareil1}. Indeed, \eqref{pareil2} will follow applying  \eqref{pareil1}  to $\pm h$. Then, \eqref{pareil3} and \eqref{pareil4} are obtained choosing $\mathcal{Q}=\{p\}.$ \\
We show the direct implication in \eqref{pareil1}. If there exists $y_0 \in D$ such that $h y_0<0$, then there exists some $\delta>0$ such that $hy<0$ for all $y \in B(y_0,\delta)$, where $B(y_0,\delta)$ is the open ball of radius $y_0$ and center $\delta$.  But by definition of $D$  there exists some $q \in \mathcal{Q}$ such that $q[Y(\cdot) \in B(y_0, \delta)]>0$, a contradiction. For the reverse implication, as in Remark \ref{EP_included_D}, 
$D$ is a closed set of $\mathbb{R}^d$
such that $q[ Y(\cdot) \in {D}] =1$ for all $q \in \mathcal{Q}$. As $Y(\cdot)$ is ${\cal B}_c(\Omega)$-measurable, $\{ Y(\cdot) \in {D}\} \in {\cal B}_c(\Omega)$ and $hY(\cdot) \geq 0$  $\mathcal{Q}$-q.s.

We prove (iv). If $0 \notin \mathrm{Ri}(\mathrm{Conv}(D)),$ (i) implies that there exists some $h^* \in \mbox{Aff}(D)$, $h^* \neq 0$ such that $h^* y \geq 0$ for all $y \in D$ or equivalently $h^* Y(\cdot) \geq 0$ $\mathcal{Q}$-q.s. using \eqref{pareil1}. As $NA(\mathcal{Q})$ holds true, $h^* Y(\cdot) = 0$ $\mathcal{Q}$-q.s. or $h^* y = 0$ for all $y \in D$ using \eqref{pareil2}. Thus, 
$$h^* \in D^\perp:=\bigl\{h\in\mathbb{R}^d: h y=0, \forall y \in D\bigr\}=(\mathrm{Aff}(D))^\perp$$ 
and also $h^* \in (\mathrm{Aff}(D))^\perp \cap \mathrm{Aff}(D)$. So, $h^*=0$ and we get a contradiction. 
Finally, (v) follows from (iv) choosing $\mathcal{Q}=\{p\}.$
\end{proof}

The following proof is partially inspired from \cite[Lemma 2.2]{bayraktar2017arbitrage} which gives the existence of some $\hat{p}  \in \mathcal{Q}$ such that $NA(\hat{p} )$  holds true and $\mbox{Aff}(E(\hat{p} ))=\mbox{Aff}(D)$. 
\begin{proposition}[Construction of $P^*$ in the one-period case]\label{prop_construction_P_star} Assume that $\mathcal{Q}$ is nonempty and convex and that the one-period $NA(\mathcal{Q})$ condition holds. Then there exists some $\hat{p}  \in Q $ such that $0 \in \mathrm{Ri}(\mathrm{Conv}(E(\hat{p} ))) $ and $\mathrm{Aff}(E(\hat{p} )) = \mathrm{Aff}(D) $.
\end{proposition}
\begin{proof}
 Assume that the quasi-sure one-period no-arbitrage $NA(\mathcal{Q})$ holds true and that $\mathcal{Q}$ is nonempty and convex. As  $0\in \mathrm{Aff}(D)$ (see Lemma~\ref{lemma_affine_hull_linear}), $\mathrm{Aff}(D)$ is a linear subspace of $\mathbb{R}^d$. We denote for all $q\in\mathcal{Q}$, 
 $$N(q):=\{h\in\mathbb{R}^d:hY(\cdot)=0\;q\mbox{-a.s}\} \mbox{ and } N(\mathcal{Q}):=\{h\in\mathbb{R}^d:hY(\cdot)=0\;\mathcal{Q}\mbox{-q.s}\}.$$ Then, using  Lemma~\ref{lemma_affine_hull_linear},
    \begin{eqnarray*} \label{la}
    \mathrm{Aff}(D)^\perp:=\bigl\{h\in\mathbb{R}^d: h y=0, \forall y \in \mathrm{Aff}(D)\bigr\}=D^\perp=N(\mathcal{Q}).
    \end{eqnarray*}
Let \begin{equation}\label{defqbar}
\bar{\mathcal{Q}}=\bigl\{ p\in  \mathcal{Q}:  {p}[hY(\cdot)< 0]>0, \forall h\in \mathrm{Aff}(D)\cap S(0,1)\bigr\},
    \end{equation}
with $S(0,1):=\{x\in\mathbb{R}^d: |x|=1\}$ (recall that $|x|$ is the Euclidian norm of  $x\in \mathbb{R}^d$). \\
 \noindent  {\it Step 1: $\bar{\mathcal{Q}}$ is nonempty.}\\
Let $h\in \mathrm{Aff}(D)\cap S(0,1)$.  There exists $p_h\in \mathcal{Q}$ such that $p_h[hY(\cdot)< 0]>0$. If not, then $hY(\cdot)\geq 0 \; \mathcal{Q}\mbox{-q.s.}$ and $NA(\mathcal{Q})$ implies that $hY(\cdot)= 0 \; \mathcal{Q}\mbox{-q.s.}$, which means that $h\in N(\mathcal{Q})=\mathrm{Aff}(D)^\perp$. Thus, $h\in\mathrm{Aff}(D)\cap \mathrm{Aff}(D)^\perp =\{0\}$. This is impossible because $|h|=1$. 
Furthermore,  
there exists $\epsilon_h>0$  such that 
    \begin{equation}\label{equation_proof_construction_P_star}
        p_h[h'Y(\cdot)<0]>0 \; \text{for all} \; h'\in  B(h,\epsilon_h). 
    \end{equation}
Now, using that $\mathrm{Aff}(D)\cap S(0,1)$ is compact in $\mathbb{R}^d$, one can extract a finite subcover of the open cover $\cup_{h\in \mathrm{Aff}(D)\cap S(0,1)}B(h,\epsilon_h).$ So, there exist $k\geq1$ and $h_i\in \mathrm{Aff}(D)\cap S(0,1)$ for all $i\in \{1,\ldots,k\}$ such that 
$\mathrm{Aff}(D)\cap S(0,1)\subseteq\bigcup_{i=1}^k B(h_i,\epsilon_i)$  setting $\epsilon_i=\epsilon_{h_i}.$ We associate to each $h_i$,  the probability $p_{h_i}\in \mathcal{Q}$ constructed above (i.e. satisfying \eqref{equation_proof_construction_P_star}) and we set 
$$\bar{p}:=\frac{1}{k}\sum_{i=1}^k  p_{h_i}.$$ Then, $\bar{p}\in \mathcal{Q}$ by convexity. 
    Furthermore, for all $h\in \mathrm{Aff}(D)\cap S(0,1)$, we have that $h\in B(h_j,\epsilon_j)$ for a certain $j\in \{1,\ldots,k\}$, 
 we can apply (\ref{equation_proof_construction_P_star}) for the probability $p_{h_j}$ and we get that 
    \begin{eqnarray*}\label{construction_bar_p}
        \bar{p}[hY(\cdot)<0]=\frac{1}{k}\sum_{i=1}^k  p_{h_i}[hY(\cdot)<0]\geq \frac{1}{k}p_{h_j}[hY(\cdot)<0]>0.
    \end{eqnarray*} 
So, we have proved that $\bar{\mathcal{Q}} \neq \emptyset$. \\
 \noindent  {\it Step 2: $\bar{\mathcal{Q}} \subseteq \{p \in \mathcal{Q}: NA(p) \mbox{ holds} \}. $}\\
Let $p \in \bar{\mathcal{Q}}$. Assume that $NA(p)$ does not hold. Then, there exists $\ell \in \mathbb{R}^d$ such that ${p}[\ell Y(\cdot)\geq0]=1$ but ${p} [\ell Y(\cdot)=0] \neq 1$, meaning that $\ell \notin N(\mathcal{Q})$. Then, the orthogonal projection of $\ell$ on $\mathrm{Aff}(D)$ is a nonzero vector (or else $\ell \in \mathrm{Aff}(D)^\perp = N(\mathcal{Q})$) and we can write $\ell=\ell' +\ell^\perp$ where $\ell' \in \mathrm{Aff}(D)$ and $\ell^\perp\in N(\mathcal{Q})$ are the respective orthogonal projections of $\ell$ on $\mathrm{Aff}(D)$ and on $\mathrm{Aff}(D)^\perp =N(\mathcal{Q})$. Then, ${p}[\ell Y(\cdot)<0]={p}[\ell \rq{}Y(\cdot)<0]=0$. This contradicts the fact that ${p}[\frac{\ell'}{|\ell'|}Y(\cdot)< 0]>0,$ see \eqref{defqbar}. Therefore, $NA({p})$ holds and the inclusion as claimed.\\
 \noindent  {\it Step 3: Construction of $\hat{p} \in \bar{\mathcal{Q}}$ such that $\mathrm{Aff}(E(\hat p))=\mathrm{Aff}(D)$. }\\
For all $q\in \bar{\mathcal{Q}}$,  as $NA(q)$ holds true,  Lemma~\ref{lemma_affine_hull_linear} shows that $ 0 \in \mathrm{Aff}(E(q)),$ which is thus a linear subspace of $\mathbb{R}^d$.
We also have the inclusion $E(q)\subseteq D$ (see Remark~\ref{EP_included_D}). So, $\mathrm{Aff}(E(q))\subseteq \mathrm{Aff}(D)$.  
Now, we set $\delta:$ $ \bar{\mathcal{Q}} \ni q \mapsto dim(\mathrm{Aff(E(q))})$.  As $\delta(\bar{\mathcal{Q}})$ is a nonempty subset of $\{0,\dots,d\}$, $m:=\max_{\bar{\mathcal{Q}}} \delta$ is attein by  some $\hat p\in\bar{\mathcal{Q}}$ and we have that 
$$\delta(\hat p)=m=\max_{\bar{\mathcal{Q}}}\delta=dim\Bigl(\mathrm{Aff}\bigl(E(\hat p)\bigr)\Bigr) \leq dim \bigl(\mathrm{Aff}(D)\bigr).$$
Using Lemma~\ref{lemma_affine_hull_linear},
    \begin{eqnarray*}\label{lap}
    \mathrm{Aff}\bigl(E(\hat p)\bigr)^\perp:=\bigl\{h\in\mathbb{R}^d: h y=0, \forall y \in \mathrm{Aff}\bigl(E(\hat p)\bigr)\bigr\}=\bigl(E(\hat p)\bigr)^\perp=N(\hat p).
    \end{eqnarray*}
Now, we prove that  $\mathrm{Aff}(E(\hat p))=\mathrm{Aff}(D)$. Else,  suppose that $\mathrm{Aff}(E(\hat p))\subsetneq\mathrm{Aff}(D).$ Let  $d_D:=\dim \mathrm{Aff}(D).$  
  First, we show that 
  $$[\mathrm{Aff}(D)\setminus\mathrm{Aff}(E(\hat{p}))]\bigcap N(\hat{p}) \neq \emptyset.$$ 
Let $B:=(b_1,\dots,b_{d})$ be an orthonormal basis  of $\mathbb{R}^d$, adapted to the decomposition of $\mathbb{R}^d=\mathrm{Aff}(E(\hat{p}))\oplus N(\hat{p})$, which $m$ first vectors make a basis of $\mathrm{Aff}(E(\hat{p}))$, and which $d_D$ first vectors make a basis of $\mathrm{Aff}(D)$. Then, $m<d_D$. We consider $b_{m+1}$. We have that $b_{m+1}\in[\mathrm{Aff}(D)\setminus\mathrm{Aff}(E(\hat{p}))]\cap N(\hat{p})$. Indeed, remember that $b_{m+1}\in\mathrm{Aff}(D)$. Moreover, let $\ell:=\sum_{i=1}^m \mu_ib_i \in\mathrm{Aff}(E(\hat{p}))$, where $\mu_1,\dots,\mu_m\in\mathbb{R}$, then $b_{m+1}\ell=\sum_{i=1}^m \mu_ib_{m+1}b_i=0.$ Thus,  $b_{m+1}\in\mathrm{Aff}(E(\hat{p}))^\perp = N(\hat{p})$. Finally, $b_{m+1}\notin \mathrm{Aff}(E(\hat{p}))$, else  
$b_{m+1}\in \mathrm{Aff}(E(\hat{p}))\cap\mathrm{Aff}(E(\hat{p}))^\perp = \{0\}$, but $b_{m+1}\neq 0$. We set $h^*:=b_{m+1}\in[\mathrm{Aff}(D)\setminus\mathrm{Aff}(E(\hat{p}))]\cap N(\hat{p}).$ Note again that $h^*\neq 0$. 

Now, as $h^*\in \mathrm{Aff}(D)$ and  $\mathrm{Aff}(D)\cap\mathrm{Aff}(D)^\perp=\{0\}$, we have that $h^*\notin \mathrm{Aff}(D)^\perp=N(\mathcal{Q})$, which means that there exists $q^*\in\mathcal{Q}$ such that $q^*[h^*Y(\cdot)\neq 0]>0$.  
We set $p'=\frac{\hat{p}+q^*}{2}$. By convexity $p'\in \mathcal{Q}$. Remark that $N(p')\subseteq N(\hat{p})$. Indeed, let $\ell \in N(p')$, then $\frac{\hat{p}+q^*}{2}[\ell Y(\cdot)=0]=1$, and necessarily $\hat{p}[\ell Y(\cdot)=0]=1$ as well, thus $\ell \in N(\hat{p})$. Moreover, $N(p')\subsetneq N(\hat{p})$. Indeed,  $h^*\in N(\hat{p})$ and as
$$p'[h^*Y(\cdot)\neq0]=\frac{\hat{p}+q^*}{2}[h^*Y(\cdot)\neq0]\geq \frac{q^*}{2}[h^*Y(\cdot)\neq0]>0$$ we have that $h^*\notin N(p').$  
So,  $N(p')\subsetneq N(\hat{p})$ and $\mathrm{Aff}(E(\hat p))\subsetneq \mathrm{Aff}(E(p')).$ Thus, $\delta (\hat p) <\delta(p')$, a contradiction to the maximality of $\hat{p}$ for $\delta$. So, $\mathrm{Aff}(E(\hat p))=\mathrm{Aff}(D)$. \\
 \noindent  {\it Step 4: conclusion.}\\
In Step 3, we have construct $\hat{p} \in \bar{\mathcal{Q}}$, such that $\mathrm{Aff}(E(\hat p))=\mathrm{Aff}(D)$.  
As $\hat{p} \in \bar{\mathcal{Q}},$ by Step 2, $NA(\hat{p})$ holds true and Lemma \ref{lemma_affine_hull_linear} implies that $0\in \mathrm{Ri}(\mathrm{Conv}(E(\hat{p}))$, which concludes the proof. 
\end{proof}

\subsection{Properties of projective sets and functions}
We present key properties of projective sets and functions used in our proofs.
\begin{proposition}[Properties of projective sets and  functions]\label{prop_properties_projective_functions}

Let $X, Y$ and $Z $ be Polish spaces.
\begin{itemize}
        \item[(i)] The sequence $(\Delta_n^1(X))_{n \geq 1}$ is a nondecreasing sequence of sigma-algebras.
    \item[(ii)] The class $\mathbf{P}(X)$ is closed under complements, finite unions and finite intersections. 
    If $A\in\mathbf{P}(X\times Y)$, then $\mathrm{proj}_X(A)\in\mathbf{P}(X)$, while  if $A\in\Sigma_n^1 (X\times Y)$ for some $n\geq 1$, then $\mathrm{proj}_X(A)\in\Sigma_n^1 (X).$
    \item[(iii)] Let $n\geq 1$. We have that $\Delta_n^1(X)\times\Delta_n^1(Y)\subseteq\Delta_n^1(X\times Y),$  $\mathbf{P}(X)\times\mathbf{P}(Y)\subseteq\mathbf{P}(X\times Y)$ and 
    \begin{eqnarray}
\Sigma_n^1(X)  \subseteq \Delta_{n+1}^1(X). \label{eq(v)_pj}
\end{eqnarray}
    \item[(iv)] Let $f: X\to \mathbb{R}^p$ and $g : X \to \mathbb{R}^p$  for some $p\geq 1$. 
    If $f$ and $g$ are projective functions, then $-f$, $fg$ and $f+g $ are also projective.
        \item[(v)]  For all $n\geq 0$, let $f, f_n, g : X \to \mathbb{R}\cup \{-\infty,+\infty\}$.
Let $p \geq 1$. Assume that $f$, $f_n$ and $g$ are $\Delta_p^1(X)$-measurable for all $n\geq 0$. Then,  $f+g$, $-f$, $\min(f,g)$, $\max(f,g)$, $\inf_{n\geq 0} f_n$, $\sup_{n\geq 0} f_n$ are $\Delta_p^1(X)$-measurable. Now, if $f$  and $g$ are projective functions, then $f+g$, $-f$, $\min(f,g)$ and $\max(f,g)$ are also projective.
     \item[(vi) ]       Let $g: D \to Y$ and $f : E \to Z$ where $D\subseteq X$ and $g(D)\subseteq E \subseteq Y$.   Assume that $f$ is $\Delta_p^1(Y)$-measurable and that $g$ is $\Delta_q^1(X)$-measurable for some $p, q\geq 1$. Then, $f\circ g$ is $\Delta_{p+q}^1(X)$-measurable. Assume that $f$ and $g$ are projective functions. Then, $f\circ g$ is also projective.
    \item[(vii) ] Let $h : X\times Y \to Z.$ Assume that $h$ is a projective function. Then $h(x,\cdot) : y \mapsto h(x,y)$ is projective for all $x\in X$, while $h(\cdot,y) : x \mapsto h(x,y)$ is projective for all $y\in Y$.
 \item[(viii)] Let $f : X \to Z$ and $h : X  \times Y \to  Z$ be defined by $h(x,y) := f(x)$ for all $(x,y)\in X\times Y$. If $f$ is $\Delta_p^1(X)$-measurable for some $p\geq 1$, then $h$ is $\Delta_p^1(X \times Y)$-measurable.  If $f$ is a projective function, then $h$ is also projective.
\end{itemize}
\end{proposition}
\begin{proof}
    Items (i) to (iii) are proved by applying \cite[Proposition 1]{carassus2024}. Note that for the projection properties in (ii), we choose the direct image with the Borel function $f:=\mathrm{proj}_X$ in \cite[Proposition 1 (i) and (vi)]{carassus2024}. Then, (iv) is proved in \cite[Lemma 4]{carassus2024}, (v)  in \cite[Proposition 4]{carassus2024}, 
  (vi)  in \cite[Proposition 3]{carassus2024}, (vii)  in \cite[Corollary 2]{carassus2024} and, finally, (viii) follows from  \cite[Lemma 2]{carassus2024}. 
\end{proof}
\begin{proposition}[Integral of projective functions] 
Assume the (PD) axiom. Let $X$ and $Y$ be Polish spaces. Let $f : X \times Y \to \mathbb{R}\cup\{-\infty,+\infty\}$ and let $q$ be a stochastic kernel on $Y$ given $X$. Let $\lambda : X  \to \mathbb{R}\cup\{-\infty,+\infty\}$ be defined by $$\lambda(x):=\int_{-} f(x,y)q(dy|x).$$
\noindent (i) Assume that $X \ni x\mapsto q(\cdot|x) \in \frak{P}(Y)$ is $\Delta_r^1(X)$-measurable for some $r\geq 1$ and that $f$ is $\Delta_p(X\times Y)$-measurable for some $p\geq 1$. Then, $\lambda$ is $\Delta_{p+r+2}^1(X)$-measurable.\\
\noindent (ii) Assume that $X \ni x\mapsto q(\cdot|x) \in \frak{P}(Y)$ is projectively measurable and that $f$ is a projective function. Then, $\lambda$ is also a projective function.
\label{prj_cvt_pj}
\end{proposition}
\begin{proof}
    This is exactly  \cite[Proposition 10]{carassus2024}. 
\end{proof}
\subsection{Projective measurability of portfolio values and of the supports.}

We now prove the measurability of the portfolio values and of the supports. 

\begin{lemma}[Projective measurability of portfolio values]\label{lemma_portfolio_measurability}
    Assume Assumption~\ref{assumption_prices}. For all $t\in\{1,\dots,T\}$, $x\in \mathbb{R}$ and $\phi \in  \Phi$,  $ \Omega^{t} \ni \omega^t \mapsto V_t^{x,\phi}(\omega^t)$ is $\mathbf{P}(\Omega^t)$-measurable, and for all $\omega^{t-1}\in \Omega^{t-1}$, $ \Omega_{t} \ni \omega_{t}  \mapsto V_{t}^{x,\phi}(\omega^{t-1},\omega_{t})$ is $\mathbf{P}(\Omega_{t})$-measurable. 
\end{lemma}

\begin{proof}
Let $t \in \{1, \dots, T\} $, $V_t^{x, \phi} = x + \sum_{s=1}^{t} \phi_s\Delta S_s$. Let $ s\in \{1,\ldots,t\}$. We have that $\phi_{s} $ is $\mathbf{P}(\Omega^{s-1}) $-measurable by definition of $\Phi$, and 
also that $ \Omega^t \ni \omega^t \mapsto \phi_{s}(\omega^{s-1})$ is $\mathbf{P} (\Omega^t) $-measurable by Proposition~\ref{prop_properties_projective_functions} (viii). 
By Assumption~\ref{assumption_prices},  $S_s $ is $\mathbf{P} (\Omega^s) $-measurable, Proposition~\ref{prop_properties_projective_functions} (viii) and (iv)  shows that  $\Omega^t \ni \omega^t \mapsto \phi_s(\omega^{s-1}) (S_{s}(\omega^{s}) -S_{s-1}(\omega^{s-1}))$ is $\mathbf{P}(\Omega^t) $-measurable, and then that $ \Omega^t \ni \omega^t \mapsto V_t^{x,\phi}(\omega^t)$ is $\mathbf{P}(\Omega^t)$-measurable.
Now, Proposition~\ref{prop_properties_projective_functions} (vii) shows that $\Omega_{t} \ni \omega_{t}  \mapsto V_{t}^{x,\phi}(\omega^{t-1},\omega_{t})$ is $\mathbf{P}(\Omega_{t})$-measurable for all $\omega^{t-1}\in \Omega^{t-1}$.
\end{proof}

The following proposition generalizes \cite[Lemma 2.6]{blanchard2020}   using similar ideas as in \cite[Proposition 12]{CF24}. 
\begin{proposition}[Projective measurability of the supports] \label{lemma_measurability_sets}
Assume the (PD) axiom and let \textit{Assumptions}~\ref{assumption_prices} \textit{and} \ref{assumption:graph_Q} hold true. Let $0 \leq t \leq T-1 $ be fixed.  
\begin{itemize}
\item[(i)]  The random sets $E^{t+1},  \ \overline{\operatorname{Conv}}(E^{t+1}), \ \operatorname{Aff}(E^{t+1}),  \ \operatorname{Ri} (\overline{\operatorname{Conv}}(E^{t+1}) ) $ are nonempty, closed-valued and $\Delta^1_n (\Omega^t \times \mathfrak{P}(\Omega_{t+1})) $-measurable for some $n \geq 1$, and thus also $ \mathbf{P}(\Omega^t \times \mathfrak{P}(\Omega_{t+1})$-measurable.     
\item[(ii)] Let $P \in \mathcal{Q}^T.$ The random sets $D_P^{t+1},  \ \overline{\operatorname{Conv}}(D^{t+1}_P),  \ \operatorname{Aff}(D^{t+1}_P),  \ \operatorname{Ri} (\overline{\operatorname{Conv}}(D^{t+1}_P) )  $ are nonempty, closed-valued and $\Delta^1_{m} (\Omega^t) $-measurable for some $m \geq n+1$,  and thus   also 
   $\mathbf{P}(\Omega^t) $-measurable. 
\item[(iii)] The random sets $D^{t+1},  \ \overline{\operatorname{Conv}}(D^{t+1}),  \ \operatorname{Aff}(D^{t+1}),  \ \operatorname{Ri} (\overline{\operatorname{Conv}}(D^{t+1}) ) $ are nonempty, closed-valued and $\Delta^1_q (\Omega^t) $-measurable for some $q \geq 1 $, and thus   also
   $\mathbf{P}(\Omega^t) $-measurable.

\end{itemize}
\end{proposition}
\begin{proof}
Recall that $0 \leq t \leq T-1 $ is fixed.  Fix also some open set $O \subseteq \mathbb{R}^d$. \\
{\it Proof of (i).}\\
First, we show that 
$
\Omega^t \times \mathfrak{P}(\Omega_{t+1}) \ni (\omega^t, p)  \mapsto p[ \Delta S_{t+1}(\omega^t, \cdot) \in O ]
$
is $\Delta^1_n (\Omega^t \times \mathfrak{P}(\Omega_{t+1}) )$-measurable for some $n\geq 1$.  
Assumption~\ref{assumption_prices} and Proposition \ref{prop_properties_projective_functions} (iv) and (viii) imply that $\Omega^{t+1} \ni \omega^{t+1}  \mapsto \Delta S_{t+1}(\omega^{t+1})$
is $\mathbf{P}(\Omega^{t+1}) $-measurable and thus   $\Delta^1_r (\Omega^{t+1}) $-measurable, for some $r \geq 1.$  
We apply Proposition \ref{prj_cvt_pj}  to the stochastic kernel $q$ defined by $q(d\omega_{t+1}|(p,\o^t))=p(d\omega_{t+1})$, which is Borel 
and thus $\Delta^1_1(\mathfrak{P}(\Omega_{t+1})\times \O^t) $-measurable and the function $f$ defined by $f(p,\o^t,\omega_{t+1})=\mathbf{1}_{\{\Delta S_{t+1}(\omega^t, \omega_{t+1}) \in O\}}$ which is $\Delta^1_r (\mathfrak{P}(\Omega_{t+1})\times \O^{t+1}) $-measurable (see Proposition \ref{prop_properties_projective_functions} (viii) again). Thus, 
$$
 \mathfrak{P}(\Omega_{t+1})  \times \Omega^t \ni (p,\omega^t) \mapsto\int_{\Omega_{t+1}} \mathbf{1}_{\{\Delta S_{t+1}(\omega^t, \omega_{t+1}) \in O\}} \, p(d\omega_{t+1})
$$
is $\Delta^1_{r+3} (\mathfrak{P}(\Omega_{t+1})\times \O^{t}) $-measurable. As $(\omega^t,p) \mapsto (p,\omega^t)$ is $\Delta^1_1( \O^t \times\mathfrak{P}(\Omega_{t+1})) $-measurable, we get that $
 \Omega^t \times \mathfrak{P}(\Omega_{t+1}) \ni (\omega^t, p)  \mapsto p[ \Delta S_{t+1}(\omega^t, \cdot) \in O ] 
$ is $\Delta^1_{r+4} ( \O^t \times\mathfrak{P}(\Omega_{t+1})) $-measurable (see Proposition \ref{prop_properties_projective_functions} (vi)). 
It follows that 
\begin{align*}
    A & :=\Big\{ (\omega^t, p) \in \Omega^t \times \mathfrak{P}(\Omega_{t+1}): \ E^{t+1} (\omega^t, p) \cap O \neq \emptyset \Big\} \\
    &= \Big\{ (\omega^t, p) \in  \Omega^t \times \mathfrak{P}(\Omega_{t+1}): \ p \big[\Delta S_{t+1} (\omega^t, .) \in O \big] > 0 \Big\} \in   \Delta^1_{r+4} (\Omega^t \times \mathfrak{P}(\Omega_{t+1})). 
\end{align*}
So, we have proved the $\Delta^1_{r+4} (\Omega^t \times \mathfrak{P}(\Omega_{t+1})) $-measurability of $E^{t+1} $. Applying \cite[Proposition 14.2, Exercise 14.12]{RockafellarWets1998} and \cite[Lemmata 5.2 and 5.7]{Artstein} 
in the measurable space $(\Omega^t \times \mathfrak{P}(\Omega_{t+1}), \Delta^1_n (\Omega^t \times \mathfrak{P}(\Omega_{t+1}))) $ with $n = r+4 $, proves that $\overline{\text{Conv}}(E^{t+1}),$ $\text{Aff}(E^{t+1}) $ and $\operatorname{Ri} (\overline{\operatorname{Conv}}(E^{t+1}) ) $  are $\Delta^1_n (\Omega^t \times \mathfrak{P}(\Omega_{t+1})) $-measurable. So, we also obtain that $E^{t+1}, \, \overline{\operatorname{Conv}}(E^{t+1}), \operatorname{Aff}(E^{t+1})$ and $\operatorname{Ri} (\overline{\operatorname{Conv}}(E^{t+1}) ) $  are $ \mathbf{P}(\Omega^t \times \mathfrak{P}(\Omega_{t+1}))$-measurable. \\
{\it Proof of (ii).}\\
Let     $P:= q_1 \otimes \dots \otimes q_T \in \mathcal{Q}^T.$ Recalling Remark~\ref{EP_included_D}, then for all $t\in \{1,\dots,T-1\}$ and $\omega^t\in \Omega^t$, $D^{t+1}_P(\omega^t)=E^{t+1}(\omega^t,q_{t+1}(\cdot\mid \omega^t))$.
We have that
\begin{align*}
    \big\{ \omega^t\in  \Omega^t: \ D^{t+1}_P (\omega^t) \cap O \neq \emptyset \big\} 
    &= \big\{ \omega^t \in  \Omega^t: \ \exists p \in \mathfrak{P}(\Omega_{t+1}), \ q_{t+1}(\cdot| \omega^t) = p, \ E^{t+1} (\omega^t, p) \cap O \neq \emptyset \big\} \\
    &= \text{proj}_{\Omega^t} \Bigl[A \cap \big\{ (\omega^t, p) \in \Omega^t \times \mathfrak{P}(\Omega_{t+1}): \ q_{t+1}(\cdot| \omega^t) = p\big\}\bigr].
\end{align*}
As $q_{t+1} \in SK_{t+1}$ and $p \in \mathfrak{P}(\Omega_{t+1}),$ we have that $\Omega^t \times \mathfrak{P}(\Omega_{t+1})  \ni (\omega^t, p) \mapsto q_{t+1}(\cdot| \omega^t) - p \in \mathfrak{P}(\Omega_{t+1})$ is 
$\mathbf{P}(\Omega^t \times \mathfrak{P}(\Omega_{t+1})) $-measurable, see Proposition \ref{prop_properties_projective_functions} (iv) and (viii), and thus 
$\Delta^1_{n\rq{}} (\Omega^t \times \mathfrak{P}(\Omega_{t+1})) $-measurable for some $n\rq{}\geq 1$. Thus, 
$$A \cap \{ (\omega^t, p) \in \Omega^t \times \mathfrak{P}(\Omega_{t+1}): \ q_{t+1}(\cdot| \omega^t) = p \} \in \Delta^1_{m-1} (\Omega^t \times \mathfrak{P}(\Omega_{t+1})) \subseteq \Sigma^1_{m-1} \bigl(\Omega^t \times \mathfrak{P}(\Omega_{t+1})\bigr), $$ 
where $m=\max(n,n\rq{})+1,$ see Proposition \ref{prop_properties_projective_functions} (i). It follows that 
$$
\big\{ \omega^t\in \Omega^t: \ D^{t+1}_P (\omega^t) \cap O \neq \emptyset \big\} =\text{proj}_{\Omega^t} \Bigl[A \cap \big\{ (\omega^t, p) \in \Omega^t \times \mathfrak{P}(\Omega_{t+1}): \ q_{t+1}(\cdot| \omega^t) = p\big\} \Bigr]\in \Sigma^1_{m-1} (\Omega^t ) \subseteq \Delta^1_{m} (\Omega^t ), 
$$
see Proposition \ref{prop_properties_projective_functions} (ii) and \eqref{eq(v)_pj}. 
So, we have proved the $\Delta^1_{m} (\Omega^t) $-measurability of $D^{t+1}_P $. Similarly, by \cite[Proposition 14.2, Exercise 14.12]{RockafellarWets1998} and \cite[Lemmata 5.2 and 5.7]{Artstein}, but this time in the measurable space $(\Omega^t , \Delta^1_{m} (\Omega^t)) $, we obtain that $\overline{\operatorname{Conv}}(D^{t+1}_P), \operatorname{Aff}(D^{t+1}_P)$ and $\operatorname{Ri} (\overline{\operatorname{Conv}}(D^{t+1}_P) )$ 
are $\Delta^1_{m} (\Omega^t) $-measurable, and thus $\mathbf{P}(\Omega^t) $-measurable.\\
{\it Proof of (iii).}\\
Proposition 12 in \cite{CF24} proves that there exists $q \geq 1 $ such that $D^{t+1} $ is $\Delta^1_{q} (\Omega^t) $-measurable. So again, applying \cite[Proposition 14.2, Exercise 14.12]{RockafellarWets1998} and \cite[Lemmata 5.2 and 5.7]{Artstein} in the measurable space $(\Omega^t , \Delta^1_{q} (\Omega^t)) $, we get that  
$\overline{\operatorname{Conv}}(D^{t+1}), $  $\operatorname{Aff}(D^{t+1}) $ and $\operatorname{Ri} (\overline{\operatorname{Conv}}(D^{t+1}) )$ are $\Delta^1_{q} (\Omega^t) $-measurable, and thus $\mathbf{P}(\Omega^t) $-measurable.

\end{proof}

\subsection{Section of jointly measurable sets}

Let $\Omega$ and $\tilde{\Omega}$ be two Polish spaces and suppose that the set-valued mapping $\mathcal{P} : \Omega \twoheadrightarrow \mathfrak{P}(\tilde{\Omega})$ is nonempty-valued. Recall that $SK$ is the set of stochastic kernels such that $q(\cdot \mid \omega)$ is a probability measure in $\mathfrak{P}(\tilde{\Omega})$ for all $\omega \in \Omega$ and $\Omega \ni \omega \mapsto q(\cdot \mid \omega) \in \mathfrak{P}(\tilde{\Omega})$ is projectively measurable.  Let  $\mathcal{R} \subseteq \mathfrak{P}(\Omega)$, we set

\[
\mathcal{Q} := \bigl\{ R \otimes q  : \, R \in \mathcal{R}, \, q \in SK,  \, q(\cdot \mid \omega) \in \mathcal{P}(\omega) \; \forall \omega \in \Omega \bigr\}.
\]

In the quasi-sure literature, it is necessary to prove that if  $\Xi:\Omega \times \tilde{\Omega}\to \mathbb{R}$ satisfies $\Xi \geq 0$ $ \mathcal{Q}$-q.s., then there exists some $\mathcal{R}$-full measure set, with the right measurability, such that for all $\o$ in this set, $\Xi(\o,\cdot) \geq 0$ $ \mathcal{P}(\o)$-q.s. This is Corollary \ref{corollary_section_B_omega}. It is based on Lemma \ref{lemma_section_B_omega}, which generalizes Lemma A.1 of \cite{Carassus25} to the projective setup. The proofs are very similar. The main difference is in the proof of the measurability of $\lambda$, which is simpler in the projective setup. 

\begin{lemma}[Section of jointly measurable sets] \label{lemma_section_B_omega}
    Assume the (PD) axiom. Assume that $Graph(\mathcal{P})\in \mathbf{P}(\Omega \times \tilde{\Omega})$.
    Fix $\bar{B} \in \mathbf{P}(\Omega \times \tilde{\Omega})$. For $\omega \in \Omega$, we denote by $\bar{B}_{\omega}$ the section of $\bar{B}$ along $\omega$, that is
    
    \[\bar{B}_{\omega} := \{\tilde{\omega} \in \tilde{\Omega} : (\omega, \tilde{\omega}) \in \bar{B}\}.\] 
    
    Then, we have \[B := \bigl\{ \omega \in \Omega : q[\bar{B}_{\omega}] = 1, \, \forall q \in \mathcal{P}(\omega) \bigr\} \in \mathbf{P}(\Omega).\]  
If furthermore $\bar{B}$ is a $\mathcal{Q}$-full measure set, then $B$ is a $\mathcal{R}$-full measure set.
\end{lemma}

\begin{proof}
Remark that $B = \{ \Lambda \geq 1 \}$, where

\[
    \Lambda(\omega) := \inf_{q \in \mathcal{P}(\omega)} q[\bar{B}_\omega].
\]

First, we prove that $\Lambda$ is $\mathbf{P}(\Omega)$-measurable. 
For that, we define,
\[
\Omega \times \mathcal{P}(\tilde{\Omega}) \ni      (\omega, q) \mapsto \lambda(\omega, q) := q[\bar{B}_\omega]=\int_{\tilde{\Omega}} \mathbf{1}_{\bar{B}_\omega}(\tilde{\omega}) q(d\tilde{\omega}).
\]
We have that the function $(\omega, q, \tilde{\omega}) \mapsto \mathbf{1}_{\bar{B}_\omega}(\tilde{\omega}) = \mathbf{1}_{\bar{B}}(\omega, \tilde{\omega}) $ is 
$\mathbf{P}(\Omega \times \mathfrak{P}(\tilde{\Omega}) \times \tilde{\Omega}) $-measurable using that $\bar{B} \in \mathbf{P}(\Omega \times \tilde{\Omega}) $ and 
Proposition \ref{prop_properties_projective_functions} (viii).
Let $p(\cdot|\cdot) : \mathcal{B}(\tilde{\Omega}) \times (\Omega \times \mathcal{P}(\tilde{\Omega})) $ be defined by  
$p(A | (\omega, q)) = q[A] $ for all $A \in \mathcal{B}(\tilde{\Omega}) $ and $(\omega, q) \in \Omega \times \mathcal{P}(\tilde{\Omega}) $.  
As $p(\cdot | (\omega, q)) = q[\cdot]$ is a probability measure on $\tilde{\Omega} $ and $(\omega, q) \mapsto p(\cdot | (\omega, q)) = q[\cdot] $ is Borel measurable 
and thus projectively measurable,  we obtain that $q \in SK $.
Recalling that $((\omega, q), \tilde{\omega}) \mapsto \mathbf{1}_{\bar{B}_\omega}(\tilde{\omega})$ is projectively measurable,  we conclude by Proposition \ref{prj_cvt_pj} that the function $(\omega, q) \mapsto \lambda(\omega, q)$ is $\mathbf{P}(\Omega \times \mathcal{P}(\tilde{\Omega}))$-measurable. For any $c \in \mathbb{R}$, we define:

\[E_c := \{ (\omega, q) \in \Omega \times \mathfrak{P}(\tilde{\Omega}) : \lambda(\omega, q) < c \} \cap \operatorname{Graph} \mathcal{P}.\]
Then, $E_c \in \mathbf{P}(\Omega \times \mathcal{P}(\tilde{\Omega}))$, see Proposition~\ref{prop_properties_projective_functions} (ii). 
Moreover, by definition of $\Lambda$ and $E_c$, we obtain that  $\{\Lambda< c \} = \operatorname{proj}_{\Omega} E_c $. Now, Proposition~\ref{prop_properties_projective_functions} (ii) again shows that $\operatorname{proj}_{\Omega} E_c \in \mathbf{P}(\Omega)$, and we conclude that $\Lambda$ is $\mathbf{P}(\Omega)$-measurable. So, $B = \{ \Lambda \geq 1 \} \in \mathbf{P}(\Omega)$.

Assume now that $\bar{B}$ is a $\mathcal{Q}$-full measure set. We prove that $B$ is a $\mathcal{R}$-full measure set.
Assume by contradiction that there exists $\tilde{R} \in \mathcal{R}$ such that $\tilde{R}[\Omega \setminus B] > 0$. 

Since $E_1 \in \mathbf{P}(\Omega \times \mathcal{P}(\tilde{\Omega})) $, we can perform measurable selection on $E_1 $ using Proposition~\ref{prop_conseq_PD_axiom}.  
So, there exists a projectively measurable stochastic kernel $\hat{q} : \operatorname{proj}_\Omega E_1 \to \mathcal{P}(\tilde{\Omega}) $, such that  
$(\omega, \hat{q}(\cdot| \omega)) \in E_1 $ for all $\omega \in \operatorname{proj}_\Omega E_1 =\{ \Lambda <1 \} =\Omega \setminus B $. 
Since $\operatorname{proj}_\Omega \operatorname{Graph}(\mathcal{P}) = \Omega $ and $\operatorname{Graph}(\mathcal{P}) \in \mathbf{P}(\Omega \times \tilde{\Omega}) $, we can also perform measurable selection on $\operatorname{Graph}(\mathcal{P}) $ using again Proposition~\ref{prop_conseq_PD_axiom},  
proving the existence of a projectively measurable stochastic kernel $\bar{q} $ such that for all $\omega \in \Omega $, $\bar{q}(\cdot \mid \omega) \in \mathcal{P}(\omega) $. We set:

\[
    \tilde{q}(\cdot| \omega) := \hat{q}(\cdot| \omega) \mathbf{1}_{\Omega \setminus B} + \bar{q}(\cdot| \omega) \mathbf{1}_B.
\]

We have that $\tilde{q} \in \mathcal{SK}$. Indeed, since $\bar{q}(\cdot|\omega )$ and $\hat{q}(\cdot|\omega ) $ are both probability measures on $\tilde{\Omega}$, $\tilde{q}(\cdot|\omega) $ is also a probability measure. Moreover,  as  $\omega \mapsto \hat{q}(\cdot|\omega) $ and $\omega \mapsto \bar{q}(\cdot| \omega) $ are projectively measurable, and $B, \Omega \setminus B \in \mathbf{P}(\Omega) $, we have that $\omega \mapsto \tilde{q}(\cdot| \omega) $ is projectively measurable, see Proposition~\ref{prop_properties_projective_functions} (ii) and (v). Moreover, as for $\omega \in \Omega \setminus B$, $(\omega, \hat{q}(\cdot| \omega)) \in E_1 \subseteq \operatorname{Graph}(\mathcal{P})$, we conclude that $\tilde{q}(\cdot| \omega) \in \mathcal{P}(\omega)$ for all $\omega \in \Omega$ and that $\tilde{R} \otimes \tilde{q} \in \mathcal{Q}$. Now, we have

\begin{align*}
    \tilde{R} \otimes \tilde{q}[\bar{B}] 
    &= \int_B \int_{\tilde{\Omega}} \mathbf{1}_{\bar{B}}(\omega, \tilde{\omega}) \bar{q}(d\tilde{\omega}|\omega) \tilde{R}(d\omega) 
    + \int_{\Omega \setminus B} \hat{q}(\bar{B}_\omega| \omega) \tilde{R}(d\omega) \\
    &\leq \tilde{R}[B] + \int_{\Omega \setminus B} \lambda(\omega, \hat{q}(\cdot| \omega)) \tilde{R}(d\omega) \\
    &< \tilde{R}[B] + \tilde{R}[\Omega \setminus B] = 1.
\end{align*}

as for all $\omega \in \Omega \setminus B$, $(\omega, \hat{q}(\cdot| \omega)) \in E_1 \subseteq \{ \lambda < 1 \}$ and $\tilde{R}[\Omega \setminus B] > 0$. This contradicts the fact that $\bar{B}$ is of $\mathcal{Q}$-full measure, and we conclude that $B$ is a $\mathcal{R}$-full measure set. 

\end{proof}

\begin{corollary}[From global to local positivity]\label{corollary_section_B_omega}
    Assume the (PD) axiom. Assume that $Graph(\mathcal{P})\in \mathbf{P}(\Omega \times \tilde{\Omega})$. Let $\Xi:\Omega \times \tilde{\Omega}\to \mathbb{R}$ be a projective function. Then, there is an equivalence between:
    \begin{enumerate}
        \item[(i)]  $\Xi\geq0 \; \mathcal{Q}\mbox{-q.s.}$ 
        \item[(ii)]  There exists a projective set of $\mathcal{R}$-full-measure  $\bar{\Omega}\subseteq{\Omega}$, such that for all $\omega\in\bar{\Omega}$, $\Xi(\omega,\cdot) \geq 0 \; \mathcal{P}(\omega)\mbox{-q.s.}$
    \end{enumerate}
\end{corollary}

\begin{proof}
    To show that (i) implies (ii), we apply  Lemma~\ref{lemma_section_B_omega} to $\bar{B} = \{\Xi \geq 0\}$ and choose $\bar{\Omega}=B$. The reverse implication is obtained by Fubini\rq{}s theorem.
\end{proof}

\textbf{Data availability statement.}\\
 No datasets were generated or analysed during the current study.

\bibliography{references.bib}

\end{document}